\documentclass[11pt]{article}

\usepackage[utf8]{inputenc}
\usepackage[T1]{fontenc}

\usepackage{url}
\usepackage{algorithm}
\usepackage{algorithmic}
\usepackage{graphicx}
\usepackage{amsfonts}
\usepackage{amsmath}
\usepackage{amsthm}
\usepackage{subfig}
\usepackage{enumerate}
\usepackage{multirow}
\usepackage{eqparbox}
\usepackage{epstopdf}
\usepackage[noadjust]{cite}
\usepackage{color}
\usepackage{todonotes}
\usepackage{booktabs}
\usepackage{geometry}
\usepackage{bbm}

\usepackage[normalem]{ulem}
\usepackage{color}
\usepackage{parskip}

\newcommand{\calO}{\mathcal{O}}
\newcommand{\E}{\mathbb{E}}

\newcommand{\Real}{{\mathbbm R}}

\newcommand{\change}[2]{#2}

\definecolor{DarkBlue}{rgb}{0.1,0.1,0.7}

\definecolor{DarkGreen}{rgb}{0.1,0.7,0.1}
 \newcommand{\schange}[2]{{\color{DarkGreen} #2 \color{black}}}

\definecolor{DarkRed}{rgb}{0.7,0.1,0.1}
\newcommand{\tcomment}[1]{}

\newtheorem{theorem}{Theorem}[section]
\newtheorem{lemma}[theorem]{Lemma}

\title{\emph{NextBestOnce}: Achieving Polylog Routing despite Non-greedy Embeddings}
\author{Stefanie Roos, Thorsten Strufe \\
TU Darmstadt, CASED \\
\texttt{firstname.lastname}@cased.de}
\date{}

\begin{document}
\maketitle

Social Overlays suffer from high message delivery delays due to insufficient routing strategies. Limiting connections to device pairs that are owned by individuals with a mutual trust relationship in real life,
they form topologies restricted to a subgraph of the social network of their users. 
While centralized, highly successful social networking services entail a complete privacy loss of their users, Social Overlays at higher performance represent an ideal private and censorship-resistant communication substrate for the same purpose.

Routing in such restricted topologies is facilitated by embedding the social
graph into a metric space. 
Decentralized routing algorithms have up to date mainly been analyzed under the assumption
of a perfect lattice structure. 
However, currently deployed embedding algorithms for privacy-preserving Social Overlays cannot achieve a sufficiently
accurate embedding and hence conventional routing algorithms fail.  
Developing Social Overlays with acceptable performance hence requires better models and enhanced algorithms, which guarantee convergence in the presence of local optima with regard to the distance to the target.

We suggest to model Social Overlays as graphs embedded in $\mathbb{Z}_n^m$ with a scale-free degree distribution
with exponent $\alpha$. The inaccuracy of the embedding is measured by a parameter $C$.
We then show that our previously introduced routing algorithm \emph{NextBestOnce}
achieves an expected routing length of $\calO\left(  \log^{\alpha-1} n  \log \log n + C^3 \log n \right)$ on our Social Overlay model.
A lower bound on the performance of \emph{NextBestOnce} is given by $\Omega\left(  \log^{\alpha-1} n + C\right)$.
Furthermore, we show that leveraging information from the two-hop neighborhood, a Neighbor-of-Neighbor (NoN) modification of our algorithm achieves an expected routing length of $\calO\left(   \log^{\delta(\alpha)(\alpha-1)} n  \log \log n + C^3 \log n\right)$, where $\delta(\alpha) < 1$. Hence, NoN information can indeed be used to improve the asymptotic routing complexity
by more than a constant factor.
\section{Introduction}
\label{sec:intro}
Centralized communication platforms, such as online social networking (OSN) services\schange{and other forums}{}, concentrate data and control in one point.
Delivering all messages and published content through the centralized provider, they allow for perfect tracking and tracing of the communicating individuals.
Current, highly successful services, like for instance Facebook, prohibit encryption and hence gain full access to the exchanged content.
\change{To increase reliability and independence from the provider, as well as the privacy of communication, several decentralized approaches have been suggested recently.}
{To enhance reliability, user-control, and privacy of communication, several decentralized approaches have been suggested recently.}

The extent of decentralization of such platforms varies, depending on their trust assumptions and security objectives:
hybrid schemes of decentralized servers, like diaspora\footnote{http://www.joindiaspora.com}, replace the single centralized instance by several interconnected servers to which the users register and connect.
Though the chosen servers can still monitor messages and behavior of their users, no central entity has full access to \change{and control over}{}all data.
Further decentralization is sought by peer-to-peer OSNs \cite{buchegger09peerson,cutillo09privacy}, which fully decentralize the service provision to all participating parties.
The devices of all participants, henceforth called {\em nodes}, are interconnected in an overlay that allows for content discovery, publication, and retrieval.
Access control in these systems is enforced through encryption and key management.
Participating in the peer-to-peer overlay, the nodes accept and establish connections with arbitrary other nodes, thus disclosing their network address to strangers, which potentially crawl the network to discover and monitor the participation of individuals.
\change{Darknets, more precisely {\em Social Overlays} \cite{ClarkeSTV10,VassermanEtAl09-MCON,EvansGrothoff11-R5N}, prevent such discovery.}{ {\em Social Overlays} \cite{ClarkeSTV10,VassermanEtAl09-MCON,EvansGrothoff11-R5N}, also called Darknets, prevent such discovery by design.}
Mapping trust of individuals onto the system, they allow connections between devices only if their owners share a mutual trust relationship in real life.
Social Overlays hence evolve topologies that reconstruct subgraphs of the social network of their participating individuals.
Both explicitly (e.g. the profile and messages) as well as implicitly shared data (e.g. participation and communication patterns) consequently can be hidden from untrusted parties.

Social Overlays currently are far from efficient enough to provide acceptable social networking or similar real-time communication services, and there is a distinct lack of analytical understanding that prevents significant enhancements.
For a larger acceptance of such privacy-preserving overlays, it is necessary to \change{solve the routing problem}
{design efficient routing algorithms with guaranteed convergence}.
\change{}{Existing models for deterministic polylog routing in small-world networks assume a base
graph in from of a lattice, which is not given in Social Overlays.}
\\
Our contribution to this complex topic is 
1) a framework for analyzing routing algorithms in the described scenario, 
2) a provably polylog routing algorithm based only on information about direct neighbors and 
3) an analysis of the gain achieved \change{with additional information when}{by additionally} considering the two-hop neighborhood for routing decisions.

\subsection{Social Overlays}
Overlays in general are application layer networks.
Formally, they are represented by a graph $G=(V,E)$ of nodes $V$ and edges $E$ between nodes.
Structured peer-to-peer systems, including distributed hash-tables (DHTs), introduce a metric namespace $M$ 
and a function $dist : M \times M \to \Real$, indicating the distance of two identifiers within this namespace.
\change{Nodes joining the overlay draw a random identifier from the namespace and after bootstrapping through an arbitrary node connect to selected nodes in order to form a routing structure.
The structure is created to allow for a stateless routing, and greedy routing is applied, which strictly decreases the overall distance to the destination in each step.
It is guaranteed to converge successfully at the destination if all closest neighbors with respect to the distance function $dist$ are connected by an edge $e \in E$.
Structured peer-to-peer systems frequently create topologies forming a lattice or a tree on the namespace, for this purpose.}
{Each node $u \in V$ is mapped to an identifier $id(u) \in M$. Edges are then chosen in such a way that the standard
routing algorithm is guaranteed to converge in a polylog number of steps.}

Social Overlays, limited by the constraint of establishing connections only between devices of individuals with mutual trust, are prevented from creating such topologies.
Early approaches for Darknets, for example Turtle \cite{PopescuEtAl06-Turtle}, use flooding, and hence are aimed at rather small group sizes.
Probabilistic search has been implemented in OneSwarm \cite{IsdalEtAl10-OneSwarm}, a Darknet protocol for BitTorrent.  
Both approaches can lead to large overhead, low success rates and long routes in case of rare files and sparse topologies. 
Second-level overlays have been proposed to decrease delays and overhead:
MCON and XVine \cite{VassermanEtAl09-MCON,MittalEtAl12-XVine} hence implement structured peer-to-peer systems by connecting the closest neigbors in the namespace through tunnels on the Social Overlay.
Discovering and maintaining these tunnels under churn, however, introduces a high overhead.
They furthermore are characterized by high delays, which make them unsuitable for most social applications.  
GNUnet, an anonymous publication system with a Darknet mode, uses recursive Kademlia for routing, restricting the neighbors to trusted contacts \cite{EvansGrothoff11-R5N}. 
Consequently, routing frequently terminates in dead ends, i.e. when a node is reached without any neighbor closer to the target. 
It hence requires a high replication rate to still locate content. 
All given approaches have mainly been proposed for anonymous file-sharing with a high replication rate for popular files. 
They are not designed to provide social networking or real-time communication services.

{\em Embedding} a routing structure into the social graph has been proposed as an alternative solution to increase the efficiency of Social Overlays.
Formally, an embedding is a function $F:V\to M$ from the set of nodes into a suitable metric namespace.
Though any such function $F$ qualifies as an embedding, the aim is to approximate a routing structure, which allows an algorithm to efficiently route messages from any source $s \in V$ to any destination $t\in V$ based on local knowledge.
The ratio between the length of the routed paths compared to the length of the actual shortest paths within the overlay is commonly termed as the {\em stretch}.
The ideal case would be a {\em no-stretch} embedding, however, no algorithm is known to achieve this, so that $M$ requires only a polylog number of dimensions.
Embeddings that allow the standard routing algorithm to terminate successfully for all source-target pairs are called {\em greedy embeddings}.
\change{They achieve that all nodes are embedded, such that they share edges with their closest neighbors on the namespace.}{
They achieve that nodes share edges with those closest to them in the namespace.}
In other words: An embedding $F$ is called greedy if for all distinct node pairs $s,t \in V$, $s$ has a neighbor $v$ that is closer to $t$, i.e. $dist(F(v), F(t)) < dist(F(s),F(t))$.
Extensive research has been performed on \emph{greedy} embeddings \cite{Maymounkov06, Kleinberg07, CvetkovskiCrovella09, FluryEtAl09, EppsteinGoodrich11-SuccinctHyperEmbed, WestphalPei09, HerzenEtAl11}, especially for wireless sensor networks and Internet routing.
All these approaches share the idea of constructing a spanning tree of the graph, which then is embedded into a hyperbolic, euclidean or custom-metric space.
The resulting embedding is a greedy embedding of the complete graph as well. 
Routing along the tree is always successful, and shorter paths may be found using additional edges as short cuts.
Dynamic node participation and potentially adversarial activities, however, require costly re-computation and maintenance of spanning tree and embedding.
The central role of the root node additionally represents a perfect target for attacks, and the identifiers disclose the structure of the spanning tree. 

Several more robust and privacy preserving embeddings have been proposed to meet Darknet requirements \cite{Sandberg06, DellAmico07, SchillerEtAl11-LCM}.
Rather than creating and embedding a spanning tree, these approaches aim at embedding the social graph in
a m-dimensional lattice using periodic adjustments of the node identifiers.
In the Darknet mode in Freenet, for example, the social graph is embedded in a ring over the namespace $[0,1)$.
\change{All nodes choose identifiers from this namespace randomly when joining.
Sampling random other nodes, they subsequently collaboratively aim at increasing the accuracy of the embedding, to enhance the approximation of a ring structure over the topology.}
{All nodes choose identifiers from this namespace randomly when joining. 
Periodically, a node selects a partner sampled by a short random walk. The two nodes decide if they swap identifiers
for an increased accuracy of the embedding, i.e. a better approximation of a ring structure over the topology.}
The resulting embeddings of these approaches are inaccurate. For instance, nodes that are neighbors in the namespace not necessarily are topological neighbors.
The standard routing algorithm hence fails and has to be adapted to deal with local optima during the routing process.

Freenet \cite{ClarkeSWH00} suggests a \emph{distance-directed depth-first search} to mitigate inaccuracies in the embedding.
Messages are forwarded to the neighbor closest to the destination that has not been contacted before.
A backtracking phase starts if a node has no neighbors left to contact or repeatedly receives the same message.
This algorithm has a low performance in large networks. The current implementation of Freenet hence additionally uses information about the identifiers of the two-hop neighborhood, thus implementing a {\em Neighbor-of-Neighbor} (NoN) routing algorithm.


\subsection{Routing: Models, Algorithms, and Complexity}

When analyzing decentralized routing algorithms, the most intensively studied property is the (maximal) expected routing length.
Let $V$ be the set of nodes and $R^A(s,t)$ denote the number of steps needed to route from node $s$ to node $t$ using algorithm $A$.
The maximal expected routing length is then given by $\max_{s, t \in V} \E(R^A(s,t))$. 
The expected routing length is similarly defined as  $\frac{1}{|V|(|V|-1)} \sum_{s \neq t \in V} \E(R^A(s,t))$.

One of the first models for the analysis of routing in small-world graphs has been proposed by Kleinberg \cite{kleinberg00small}. 
Here, nodes are placed on a m-dimensional lattice.
Each node $v$ then is connected to all nodes within distance $p\geq 1$ and additionally has $q\geq 1$ long-range contacts. 
A long-range contact $u$ is chosen with probability anti-proportional to $d^r$ for some $r>0$, where $d$ is the distance of $v$ to $u$.
The routing length of the standard algorithm  with respect to the described topology model is polylog if and only if $r=m$.
These results for the case $r=m$ have been extended in various ways: 
It has been shown that the standard routing algorithm has expected routing length $\Theta\left(\log^2 n\right)$ steps.
Since the diameter is logarithmic, this is not asymptotically optimal. 
Consequently, extensions of the routing algorithm using the information of $\lceil \log n \rceil$ 
nodes in each step have been proposed, which reduce the expected routing length to  $\Theta\left(\log^{1+1/m} n\right)$ \cite{Martel03thecomplexity}.
Similar alternative routing algorithms, considering a larger neighborhood before choosing the next hop, have been discussed in \cite{Lebhar04almost, Giakkoupis11optimal}. 
Though achieving close to optimal or optimal performance, these algorithms are designed considering a constant degree distribution. 
Furthermore, they are based on
additional knowledge about the network size, which is not supposed to be known in a privacy-preserving embedding.
Closer related to the topic of Social Overlays, the standard routing algorithm has been analyzed in case  the degree of a node is chosen according to a scale-free distribution with exponent $\alpha$.
The expected routing length for directed scale-free graphs is asymptotically the same as in the original model, but in case of undirected links, it is reduced to $\calO\left( \log^{\alpha-1} n \log \log n \right)$ \cite{fraigniaud09effect}.
The case of using Neighbor-of-Neighbor (NoN) information for routing has been treated in \cite{Manku04NoN}, finding that with $\Theta(\log n)$ neighbors per node the expected routing length is asymptotically equal to the diameter.

Additionally to the routing performance, various properties of small-world models have been analyzed. 
Detailed studies on the diameter of such graphs with regard to the clustering exponent $r$ have been made \cite{Coppersmith02diameter,Martel04analyzing}.
Furthermore, an generative model on how long-range links are created by a random process modeling the movement of individuals over time has been suggested   \cite{Chaintreau08networks}.
However, all works assume an underlying lattice structure, so that each node shares an edge to those that are closest to it. 
Considering arbitrary base graphs rather than lattices leads to an expected routing length of $n^{\Omega(1)}$ \cite{Fraigniaud10searchability}. 
We assume that our embedding algorithms provide an enhanced structure and thus a lower routing complexity than using unstructured graphs.
To the best of our knowledge, the only result about local edges is that they are necessary for the connectivity of the graph \cite{Martel03thecomplexity}.
Hence, though heuristic embedding algorithms do not achieve links between nodes closest in the namespace, connectivity and routing success require some type of local connections.
We build our models considering, extending and complementing the above results. 
Our main modification lies in introducing a parameter governing the accuracy of the embedding, while at the same time guaranteeing connectivity always surely.


\subsection{Prior Work and Contributions}

In prior work, we have extended Kleinberg's model to address the expected inaccuracy of heuristic embeddings.
Nodes hence are not connected to their closest, but to nodes within a specified distance in a lattice \cite{roos12provable}. 
The accuracy of the embedding is reflected by the maximal distance $C$ between closest neighbors.
We also have shown that the Freenet algorithm does not achieve polylog routing paths and suggested \emph{NextBestOnce} \cite{roos13contribution}.

In this paper, we prove that \emph{NextBestOnce} has polylog maximal expected routing length for sufficiently accurate embeddings of social graphs. 
We model social graphs as graphs with a scale-free degree distribution with exponent $\alpha$. 
Additionally, we quantify the gain of using information about the neighbors' neighbors for routing. 
The extended algorithm \emph{NextBestOnce-NoN} is shown to have a maximal expected routing length of $\calO\left(\log^{\delta(\alpha)(\alpha-1)} n \log \log n + C^3 \log n\right)$ for $\delta(\alpha) < 1$,
whereas \emph{NextBestOnce} only achieves an expected routing length of $\Omega\left(\log^{(\alpha-1)} n + C \right)$
and $\calO\left(\log^{(\alpha-1)} n \log \log n + C^3 \log n\right)$.

The existence of local minima with regard to the distance to the target complicates the analysis decisively.
Our methodology needs to depart from the traditional analysis of routing as an integer-valued decreasing random process. 
Various techniques for probabilistically bounding the increase due to such a local minima have been exploited to provide the required proofs.

We start by precisely defining the model and introducing the routing algorithms in Section \ref{sec:model}.
Afterwards, we state our results with short sketches of the proofs in Section \ref{sec:results}.  
In Sections \ref{sec:upper} and  \ref{sec:lower} the proof for the upper and lower on \emph{NextBestOnce}
are presented. \emph{NextBestOnce-NoN} is treated in Section \ref{sec:nonbound}.
The paper is completed by a discussion of the results and their impact in Section \ref{sec:conclusion}. 

%

\section{Preliminaries}
\label{sec:model}

In this section, our graph model of a Social Overlay
, routing algorithms, and central definitions are introduced.

\subsection{Modelling inaccurate embeddings}
\label{sec:somodel}

We use a model for restricted topologies with heuristic embeddings, an extension
to Kleinberg's small-world model \cite{kleinberg00small}. Though Kleinberg's model
proposes an explanation how short paths are found in small-world networks, it is only of restricted
use with respect to Social Overlays. The main discrepancy between an embedded trust graph and 
Kleinberg's model is the underlying lattice structure of the latter. The currently employed privacy-preserving
embedding algorithms cannot achieve a greedy embedding. Rather, they result in nodes that are not
connected to those that are closest with regard to the distance of identifiers.
Our Social Overlay model provides a parameter $C$ for the accuracy of the embedding.
Furthermore, Kleinberg's model is extended to allow for arbitrary degree distributions
and undirected graphs \cite{roos13contribution}.

A graph of the class $\mathcal{D}(n, m, C, L)$ consists of $n^m$ nodes, arranged in a
$m$-dimensional hypercube of side length $n$, so nodes are given unique identifiers (IDs)
in $\mathbb{Z}_n^m$. In the following, we use the name $v$ of a node synonymously with
its identifier $id(v)=(v_1,\ldots , v_m)$. 
The distance between two nodes $u$ and $v$ is given by the Manhattan distance with wrap-around:
\begin{align*}
\label{eq:dist}
\begin{split}
dist(v,u) =  \sum_{i=1}^d \min\{|v_i - u_i|,  n - |v_i - u_i| \}
\end{split} 
\end{align*}

The parameter $C$ is a measure for the
accuracy of the embedding, and gives the maximal distance to the closest neighbor
in each principal direction.

More precisely, each node $v=(v_1,\ldots ,v_m)$ is given short-range links to neighbors $ a^v_1,\ldots ,a^v_m, b^v_1,\ldots ,b^v_m$. Here $a^v_j$ is chosen from the set 
\begin{align*}
\begin{split}
A^v_j = \{ u = (u_1,...,u_m) \in V : u_i = v_i \text{ for } i \neq j, 
 1 \leq \min \{ u_j - v_j, n + u_j - v_j \} \leq C\}.
\end{split}
\end{align*}
Analogously, $b^v_j$ is chosen from
\begin{align*}
\begin{split}
B^v_j = \{ u = (u_1,...,u_m) \in V : u_i = v_i \text{ for } i \neq j, 
 1 \leq \min \{ v_j - u_j, n + v_j - u_j \} \leq C\}. 
\end{split}
\end{align*}
The random variable $L$ governs the degree distribution, an inherent property of the trust graph.
In addition to the short-range links, long-range links are chosen in a two step process: 
\begin{enumerate}
\item choose a label $l_v \in \mathbb{N}$, distributed according to $L$, for each node $v \in V$ 
\item connect nodes $u,v$ with probability
\begin{equation} \label{eq:LRtilde} P(l(u,v) | l_u = l_1,l_v=l_2,dist(u,v)=d) = 1-e^{-\frac{l_1l_2}{d^m\gamma}} \end{equation}
where $\gamma $ is a normalization constant chosen such that 
\begin{align}
\label{eq:expectation}
2\sum_{d=1}^{n/2}\sum_{l_1=1}^\infty \left( 1-e^{-\frac{l_1}{d^m\gamma}} \right) P(L=l_1) = 1,
\end{align}
i.e. the expected number of long-range links of a node $v$ with label $l_v = 1$ is 1. 
\end{enumerate} 
Note that we abbreviate the event $\{ l(u,v)=1 \}$ by $l(u,v)$. In general, brackets indicating events are
dropped to enhance readability in later sections.
The above model has proven useful in analyzing routing alternative to a memoryless greedy approach, which
is bound to fail in case the embedding is not greedy.

Fraigniaud and Giakkoupis \cite{fraigniaud09effect} also analyzed routing in small-world networks with a scale-free degree distribution.
In their generative model, long-range links are first created as directed edges
and then the reverse edges are added.
This approach complicates an analysis of NoN routing, since
one has to distinguish in which direction edges were originally
selected.



\subsection{Routing Algorithms}
\label{sec:algos}

For deterministic routing based on a non-greedy embedding, state information is needed to avoid loops
and dead ends.  
Backtracking is used in case a node has no suitable neighbor to forward the message to. 
Furthermore, nodes are \emph{marked} when they should only be contacted for backtracking in the future. 

\begin{minipage}[ht]{.98\linewidth}
\begin{algorithm}[H]
\caption{\footnotesize NextBestOnce$^*$(Node v, Node p, ID t, Set $B$, boolean b)}
\label{algo:route}
\begin{algorithmic}[1]
{\small
\STATE \COMMENT {input: v: message holder, p: predecessor, t: target, $B$: {\em marked} nodes, b: backtracking?}
\STATE \COMMENT {\ensuremath{N_v}: neighbors of v}
\STATE \COMMENT {\ensuremath{IDS(u)}: set of identifiers associated with $u$}
\IF {id(v)  == t}
\STATE routing successful; terminate
\ENDIF
\IF {!backtrack}
\STATE v.predecessor.add(p);
\ENDIF
\STATE $S = \{ u \in N_v:  !B.contains(u) \}$
\IF {$S$ NOT EMPTY}
\STATE nextNode = $argmin_{u \in S} dist(IDS(u)),t)$ 
\STATE b = false; 
\IF {$dist(nextNode,t) \geq dist(v,t)$}
\STATE B.add(v)
\ENDIF
\ELSE
  \STATE B.add(v)
   \STATE nextNode = v.predecessor.pop(); 
    \STATE b = true; 
\ENDIF    
\IF {nextNode != null}
\STATE NextBestOnce(v, t, nextNode, B, b) 
\ELSE
\STATE routing failed; terminate
\ENDIF
} 
\end{algorithmic}
\end{algorithm}
\end{minipage}
\newline

The order by which nodes are \emph{marked} is crucial for the routing performance. 
The straight-forward approach, currently implemented in Freenet, 
is a distance-directed depth-first search, \emph{marking}
nodes the first time they are contacted. 
However, this algorithm does not achieve polylog expected routing length, as is shown in
\cite{roos13contribution}.
Consequently, \emph{NextBestOnce} was introduced, which allows a node to contact all
neighbors closer to the destination before \emph{marking} it.
\emph{NextBestOnce} has been shown to achieve
polylog maximal expected routing length for constant $C$ for a simplified version of the
above model, 
but simulations indicate that for
realistic network sizes, the performance gain in comparison to Freenet routing
is limited \cite{roos12provable, roos13contribution}.

Hence, we suggest to enhance the performance of the algorithm by using additional
information. Rather than only one identifier, each node provides a set of
identifiers. 
The extended algorithm, \emph{NextBestOnce$^*$}, is described in Algorithm \ref{algo:route}.
The input of \emph{NextBestOnce$^*$} consists of the current message holder $v$, the predecessor $p$ of $v$, the target ID $t$, 
the set $B$ of \emph{marked} nodes, and a flag $b$ indicating if the routing
is in the backtracking phase. Note that $B$ can be realized in a privacy-preserving manner,
e.g. by relying on a bloom filter, and is not decisive for the asymptotic routing length.
Each node keeps a stack of predecessors for backtracking, which are contacted if $v$
has only \emph{marked} neighbors (ll. 19-20).

If at least one neighbor is not \emph{marked}, 
$v$ selects the not \emph{marked} neighbor 
$u$, so that the distance to one of the identifiers in $IDS(u)$ is minimal (l. 12).
For \emph{NextBestOnce}, this set of identifiers only consists of the
ID of $u$, the generic version allows a node to provide multiple IDs.
Though we focus on the case that $IDS(u)$ consists of the IDs of $u$ and its neighbors,
the algorithm \emph{NextBestOnce$^*$} allows for e.g. multiple realities as well. 

After determining the next node $u$ on the path, $v$ is \emph{marked} if $u$
is at a larger distance to $t$ (l. 15) or backtracking starts (l.18). 
To guarantee termination, only one representative ID
of $u$ is considered for the decision of \emph{marking} $v$.


\subsection{Definitions and Notation}
\label{sec:notation}

In the remainder of the paper, we analyze the performance of two 
routing algorithms based on \emph{NextBestOnce*}.
The first one,  \emph{NextBestOnce} has been proposed
in \cite{roos12provable} and only uses the identifiers
of the direct neighbors, i.e. $IDs(u)=id(u)$ in Algorithm
\ref{algo:route}.
The second algorithm, \emph{NextBestOnce-NoN}, uses information about neighbors
of neighbors, i.e. $IDs(u) = \{id(u) \} \cup \bigcup_{v \in N(u)} \{ id(v) \} \}$,
where $N(u)$ is the set of neighbors of node $u$.
  

The number of hops required by \emph{NextBestOnce}, respectively 
\emph{NextBestOnce-NoN},
to find a path from source $s$ to destination s$t$ is denoted by $R^{NBO}(s,t)$ and $R^{NoN}(s,t)$, respectively.

The performance is analyzed with regard to the Social Overlay model presented above.
Labels are chosen according to a scale-free distribution $S_{\alpha}$ with exponent $2 < \alpha < 3$ 
and a maximum $\mu$, i.e.
\begin{align}
\label{eq:salpha}
P(S_{\alpha}=k) \propto \frac{1}{k^{\alpha}}, \quad k=1\ldots \mu.
\end{align}
Scale-free degree distribution are common in various complex networks, especially social
networks.
Furthermore, the set $B_{d} (v) = \{u: dist(v,u) < d \}$ contains all nodes at
distance less than $d$ of $v$.

For reasons of presentation, results are given for $m=1$ dimensions, but can
analogously be derived for multi-dimensional identifier spaces.
\section{Results}
\label{sec:results}



We present upper and lower bounds for the performance of \emph{NextBestOnce}.
The routing length increases at least linearly with $C$,   
the maximal distance to a local neighbor.
If $C$ is constant, the bounds are the same as those in \cite{fraigniaud09effect} for a small-world
model with a scale-free degree distribution and edges between all nodes within distance
1. 
This agreement is non-trivial, since it has been shown in 
\cite{roos13contribution} that straight-forward extensions to the standard routing algorithm
do not achieve polylog performance for $C > 2$.

\begin{theorem}
\label{thm:nboupper}
For a graph $G=(V, E) \in \mathcal{D}(n, 1, C, S_{\alpha})$ and two nodes $s,t \in V$ with
distance $d=dist(s,t)$,
an upper bound on the expected routing length of \emph{NextBestOnce} is given by
\begin{align}
\E(R^{NBO} (s,t)) = \calO\left(\log^{\alpha-1} d \log \log d + C^3 \log n \right).
\end{align}
The maximal expected routing length is consequently
\begin{align}
\max_{s, t \in V} \E(R^{NBO} (s,t)) = \calO\left(\log^{\alpha-1} n \log \log n + C^3 \log n \right).
\end{align} 
\end{theorem}

For the proof two routing phases are considered. First, the number of steps
to reach a node within distance $C$ of the target is bound by 
$\calO\left(\log^{\alpha-1} n \log \log n\right)$. Note that up to this point, the distance to
the target can be modeled as a monotonously decreasing random process. Hence,
the proof is essentially the same as in \cite{fraigniaud09effect}. Our 
contribution lies in bounding the remaining number of steps. This consists
of a) showing that with high probability no node at a distance exceeding
$C^2\log n$ is contacted from this point on, and b) the worst case complexity
of \emph{NextBestOnce} on a graph of size $N$ is $\calO\left(CN\right)$. The bound then follows
from applying b) to the subgraph of size $C^2\log n$. The proof is presented
in Section \ref{sec:upper}.

\begin{theorem}
\label{thm:nbolower}
For a graph $G=(V, E) \in \mathcal{D}(n, 1, C, S_{\alpha})$ with $C < \frac{1}{4}n^{1/4}$ and two nodes $s,t \in V$,
a lower bound on the expected routing length of \emph{NextBestOnce} is given by
\begin{align}
\frac{1}{n(n-1)} \sum_{s \neq t \in V}\E(R^{NBO} (s,t)) = \Omega\left(\log^{\alpha-1} n + C\right)
\end{align}
\end{theorem} 

Again, the first term follows essentially from \cite{fraigniaud09effect}, bounding
the number of steps to reach a node within distance $C$ of the target.
The second term is derived by considering that with constant probability $t$
and its only 2 local neighbors have at most one long-range link each. 
Conditioning
on this event, it can shown that again with constant probability routing needs 
at least $C/4$ steps after reaching a node within distance $C$ of $t$. 
The proof can be found in Section \ref{sec:lower}. 

This concludes our results for \emph{NextBestOnce}. The upper bound guarantees
polylog routing length as long as $C$ is polylog. 
The lower bound
provides a way to measure the gain of NoN information for routing.
We show that when considering NoN information the expected routing length is reduced by more than 
a constant factor, although the average degree is constant. 
Indeed, the routing
length $T$ of \emph{NextBestOnce} is reduced to $T^\delta$ for some $0 < \delta < 1$
in case of \emph{NextBestOnce-NoN}.

\begin{theorem}
\label{thm:non}
For a graph $G=(V, E) \in \mathcal{D}(n, 1, C, S_{\alpha})$, 
an upper bound on the maximal expected routing length of \emph{NextBestOnce-NoN} is given by
\begin{align}
\max_{s,t \in V} \E(R^{NoN} (s,t)) = \calO\left(\log^{\delta(\alpha)(\alpha-1)} n \log \log n + C^3 \log n \right) \textnormal{ for } \delta(\alpha)=1-\frac{(\alpha-2)(3-\alpha)}{\alpha}.
\end{align}
\end{theorem}

The proof is rather lengthy. Selecting parameters $0 \leq k \leq \alpha-2$, and $1/2 \leq r \leq 1$,
we bound the number of steps to reach a node within distance $e^{\log^r n}$ of $t$ by 
$\calO\left(\log^{\alpha-r - k(3-\alpha)} n \log \log n\right)$. 
The idea is to determine the probability of halving the distance in the next two steps. 
For this purpose, the probability of contacting
a nodes of degree $\log^k n$ and $\log n$ is derived.
With constant probability, the later has a neighbor at half its distance to $t$. 
From this, the above bound for the first phase can then by derived using basic results about stochastic processes.
The steps needed to cover the remaining distance
are then at most $\calO\left(\log^{r(\alpha-1)}n + C^3 \log n\right)$
by Theorem \ref{thm:nboupper}.
Afterwards, the result is obtained by finding the minimum of a two-dimensional extremal value problem with
variables $k$ and $r$.
The complete proof is presented in Section \ref{sec:nonbound}.

Our results are obtained under the assumption that the graph $G$ is connected.
This holds for $C=\calO(n^s)$ for any $s < 1$
with overwhelming probability. The proof of this follows easily from Lemma \ref{lem:logrNeigh}
and can be found in \cite{roos11analysis}.


\section{Proof of Theorem \ref{thm:nboupper}}
\label{sec:upper}
The upper bound on \emph{NextBestOnce}'s routing length is derived by dividing the routing
into two phases: the number of steps $R^{NBO}_1(s,t)$ needed to reach a node $v$ within distance $C$ of $t$ 
and the number of steps $R^{NBO}_2(s,t)$ to reach $t$ from $v$.

\begin{lemma}
\label{lem:phase1}
For a graph $G=(V, E) \in  \mathcal{D}(n, 1, C, S_{\alpha})$ and two nodes $s,t \in V$
with $d=dist(s,t)$,
the expected routing length of \emph{NextBestOnce} during the first phase is
\begin{align*}
\E(R^{NBO}_1(s,t)) = \calO\left( \log^{\alpha-1} d \log \log d\right)
\end{align*}
\end{lemma}
Since the distance to $t$ decreases by at least $1$ in each step during the
first phase, the above lemma is essentially treated in \cite{fraigniaud09effect}.
A complete proof for our slightly different model can be found in \cite{roos11analysis}.


In the following, we prove that the remaining distance is covered in 
$\calO\left( C^3 \log n \right)$
steps.
This requires two preliminary results:
It needs to be shown that there exist polylog paths between two nodes within distance $C$, 
and that these paths are found by \emph{NextBestOnce}.
Lemma \ref{lem:logrNeigh} gives the probability that two nodes are connected by
a so called \emph{greedy} path, i.e. a path $u_0,u_1,...,u_{l+1}$, so that $dist(u_i,v) < dist(u_{i-1},v)$
for $i=0,\ldots ,l$. 
Let $g(v,w)$ indicate if $v$ and $w$ are connected by a greedy path. For brevity, we write $P(g(v,w))$ rather
than $P(g(v,w) = 1)$.
Secondly, we show in Lemma \ref{lem:performance} that \emph{NextBestOnce} has routing length $\calO\left( CN\right)$ on any (sub-)graph
of order $N$.  Finally, Lemma \ref{lem:logrNeigh} is applied to show that with overwhelming probability
no node at distance $\Omega\left( C^2 \log n \right)$ to $t$ is contacted, so that the bound $\calO\left( C^3 \log n \right)$ is a 
consequence of Lemma \ref{lem:performance}. 


\begin{lemma}
\label{lem:logrNeigh} 
 For two nodes $w,v \in V$ with $dist(w,v) > C^2\log n$, the probability that $w,v$ are connected by a greedy path is
 $P(g(w,v)) = \Omega\left( 1 - \frac{1}{n}\right)$
\end{lemma}
\begin{proof}
Recall from Section \ref{sec:somodel} that each node $u$ has two short-range neighbors 
$a^u_1$ and $b^u_1$ chosen independently of each other. They are both within distance $C$ of $u$, but in opposite directions. 

For all pairs $(v,w) \in V\times V$, there is a path of short-range links of length at most $C$ originating at $v$ leading to a node within distance $C$
of $w$ and via versa. A greedy path between $v$ and $w$ exists if those two paths intersect (see Figure \ref{fig:greedypath}). 
Denote by $g^a(u_0,u_{l+1})$ the event that $u_0$ and $u_{l+1}$ are connected by a greedy path and
$a^{u_i}_1 = u_{i+1}$ for $i=0, \ldots , l$.
$g^b(u_0,u_{l+1})$ is defined analogously.
Without loss of generality, $w$ is 'above' $v$ in the namespace, 
i.e. $dist(w,v) = w-v$ mod $n$.

In the following, we 
show that
  $P(g^a(v , u)) \geq \frac{1}{C}$
for all $u \in U:= \{u \in V: dist(v,w) = dist(v,u) + dist(u,w) \}$, i.e.
all nodes in the shorter ring segment between $v$ and $w$.
If $dist(v,u) \leq C$, then $P(g^a(v ,u)) \geq P(a^v_1 = u) = \frac{1}{C}$ by the choice of short-range links. 
Otherwise, there exists a node $z$ such that $g^a(v, z)$ holds and
$dist(v,u) - C \leq dist(v,z) < dist(v,u)$.
It follows that $P(g^a(v ,u)) \geq P(a^z_1 = u) = \frac{1}{C}$.
Similarly, $P(g^b(w ,u)) \geq \frac{1}{C}$ holds for all $u \in U$. 

Because $a^u_1$ and $b^u_1$ are chosen independently, the probability that the two paths
intersect can be bounded as follows:
\begin{align*}
\begin{split}
P\left(g(v,w)\right)
&\geq P\left(\bigcup_{u \in U} \left(g^a(v,u)\cap g^b(w,u)\right) \right) \\
&\geq 1 - \left(1-\frac{1}{C^2}\right)^{C^2 \log n} \\ 
&\geq 1 - e^{-\frac{C^2log n}{C^2}} \\
&= 1 - \frac{1}{n}
\end{split}
\end{align*}
The third inequality follows from $1-x \leq e^{-x}$ for $x \in [0,1]$.
\end{proof}
As a second step, a worst-case bound on the routing length of \emph{NextBestOnce} 
is needed.

\begin{lemma} 
\label{lem:performance}
Let $G = (V,E)$ be an undirected graph that is embedded in $\mathbb{Z}_{|V|}$, so that all $v \in V$ are connected to nodes within distance $C$ in each direction.  The expected routing length of \emph{NextBestOnce} on $G$ is bounded by
\begin{align*}
\max_{s,t \in V}\E(R^{NBO} (s,t)) = \calO\left( C|V| \right)
\end{align*} 
\end{lemma}
\begin{proof}
The algorithm definitively terminates after every node has been {\em marked}. We show that in average
at least every $C$-th node is {\em marked}. 
The bound follows immediately.
First note that the maximal increase in distance per hop is $C$: Each node $u$ has a short-range link to a node $v$, so that the $ dist(u,t) < dist(v,t) \leq dist(u,t) + C$.  
$v$ is not yet {\em marked}, because a node is only {\em marked} after all neighbors closer to the destination, including the current message holder $u$, have been {\em marked}. 
Therefore, \emph{NextBestOnce} can always choose a successor within distance $C$.
The maximal path length without producing a circle is $|V|$.
Assume the algorithm produces a circle of length $l$. 
Then at least $l/C$ nodes on the circle are {\em marked}. 
To see this, recall that an increase in the distance $t$ implies that a node is declared {\em marked}. 
In case of a circle the sum of the distance changes per hop equals zero, 
so the distance is increased in at least $\frac{1}{C} = \frac{minDecrease}{maxIncrease}$ of all hops of the circle.
The maximal number of hops without circles and the maximum number of hops in circles until all nodes have been \emph{marked}
give the bound 
\begin{align*}
\max_{s,t \in V}\E(R^{NBO} (s,t)) 
\leq |V|+C|V| 
= \calO\left( C|V| \right).
\end{align*}
\end{proof}
It follows that the maximal number of steps is linear in the network size if the maximal increase to the destination is restricted by a parameter $C$ independent of $n$. 
For arbitrary graphs, the algorithm terminates after $\calO\left(n^2 \right)$ steps by  Lemma \ref{lem:performance}. 
The last two lemmata enable us to bound the complexity of 
\emph{NextBestOnce} during the second phase. 

\begin{lemma}
\label{lem:phase2}
For a graph $G=(V, E) \in  \mathcal{D}(n, 1, C, S_{\alpha})$ and two nodes $s,t \in V$,
the expected routing length of \emph{NextBestOnce} in the second phase is bounded by
\begin{align*}
\E(R^{NBO}_2(s,t)) = \calO\left( C^3 \log n\right).
\end{align*}
\end{lemma}
\begin{proof}
Denote the first node in $B_{C}(t)$ that is on the routing path by $v$.
Consider the event $A$ that no node at a distance exceeding $C^2 \log n + C$ to the set $\{ v,t\}$ is contacted during the second phase of
the routing.
Using $A^{\bot}$ for the complement of $A$, we get
\begin{align*}
\begin{split}
 \E(R^{NBO}_2(s,t)) 
 &= P(A)\cdot \E(R^{NBO}_2(s,t)| A) + (1- P(A))\cdot \E(R^{NBO}_2(s,t)| A^{\bot}) \\
 &=  P(A) \cdot \calO\left(C^3  \log n \right) + (1- P(A))\cdot \calO\left( Cn \right)
 \end{split}
\end{align*}
The last step follows from applying Lemma \ref{lem:performance} to the subgraph of size $C^2 \log n$ as well as to the
whole graph $G$. 

It remains to determine $P(A)$.
The claim holds if $v=t$. 
Otherwise, let $x^v_i$ for $i=0,...,C-1$ be the node such that $dist(x^v_i,t)= C^2\log n + i$ and $dist(x^v_i,t)> dist(x^v_i,v)$. Analogously, $x^t_i$ denotes the node such that $dist(x^v_i,t)= C^2 \log n + i$ and $dist(x^t_i,v)> dist(x^v_i,t)$. 
The set $X= \{x^u_i: u \in \{v,t \}, 0 \leq i < C \}$ consists of two sets of consecutive nodes at distance $C^2 \log n$ to $C^2 \log n + C -1$ from the set $\{ v,t\}$ (see Figure \ref{fig:setX}).
If a node at a higher distance than $C^2 \log n$ is reached after $v$, at least one node in $X$ needs
to be on the path as well, because the maximal regression per hop is bound by $C$.
Recall that \emph{NextBestOnce} \emph{marks} a node $u$ if all its neighbors closer to $t$ have been \emph{marked}.
It follows recursively that if a successor at a higher distance 
than the current node $u$ is chosen,
all nodes reachable from $u$ by paths along which the distance to $t$
decreases monotonously have been \emph{marked}. 
Consequently, a node $u$ with $dist(u,t) \geq C^2 \log n + C$ can only be on the path if
all nodes $x \in X$ do \emph{not} have a greedy path to $t$. 
\begin{figure}
\begin{minipage}[t]{0.42\linewidth}
\centering
\includegraphics[width=0.95\linewidth]{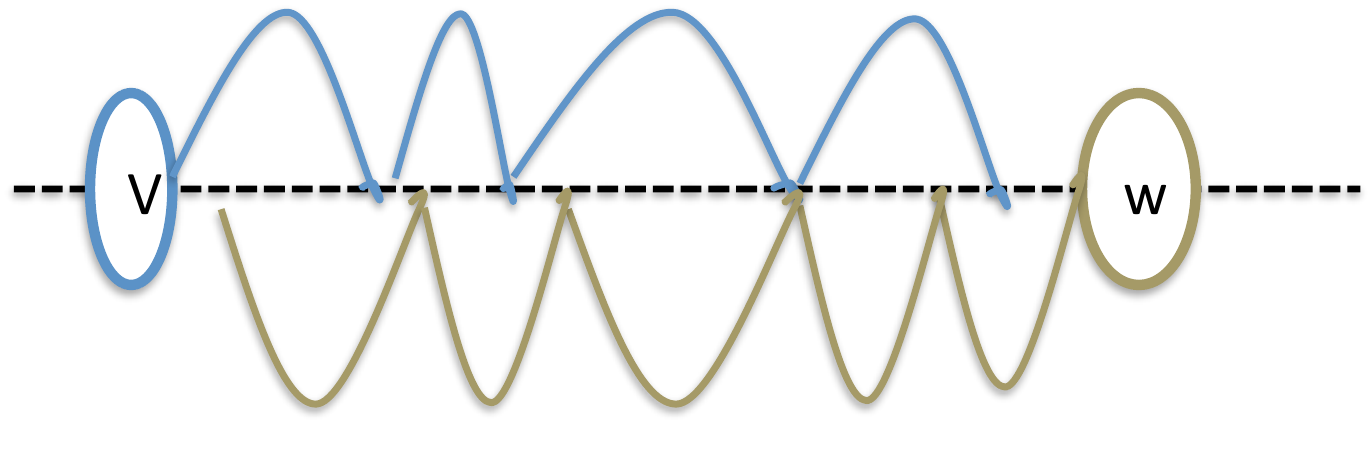}
\caption{Path of edges of length at most $C$ originating from $v$, respectively $w$. 
A greedy path between the $v$ and $w$ exists, because the two paths intersect.}
\label{fig:greedypath}
\end{minipage}
\hspace{0.5cm}
\begin{minipage}[t]{0.54\linewidth}
\centering
\includegraphics[width=0.95\linewidth]{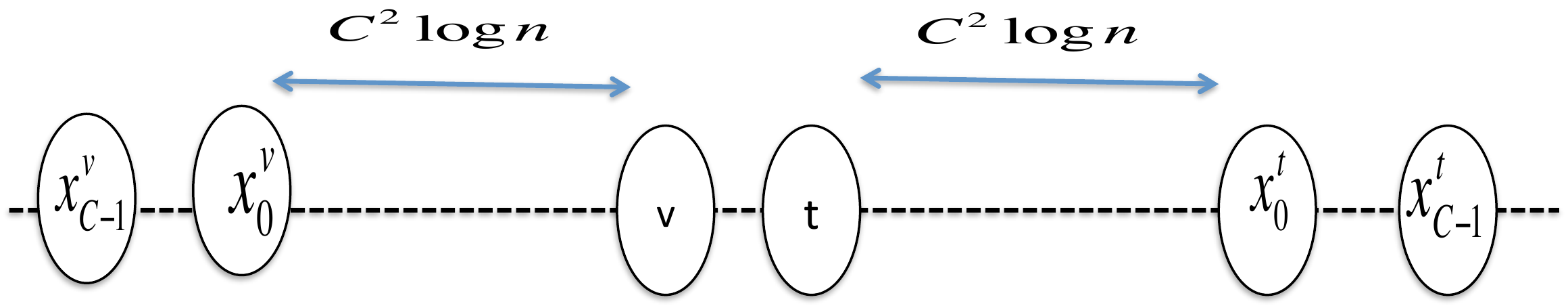}
\caption{Nodes $v$ and $t$ and the corresponding set $X$: No node at a distance exceeding 
$C^2 \log n + C$ is contacted with high probability.}
\label{fig:setX}
\end{minipage}
\vspace{-1em}
\end{figure}

Lemma \ref{lem:logrNeigh} is applied to bound $P(A)$ by  the probability that all $2C$ nodes in $X$ have a greedy path to $t$:
\begin{align*}
P(A) 
\geq P\left(\bigcap_{x \in X} g(t,x)\right) 
= \Omega\left( \left(1 - \frac{1}{n}\right)^{2C}\right)
=\Omega\left(1 - \frac{2C}{n}\right)
\end{align*}
The last step holds since $(1-x)^k \geq 1-kx$ for $0 < x < 1$ and $k > 1$.
Finally, we get
\begin{align*}
\begin{split}
\E(R^{NBO}_2(s,t)) 
=  \Omega\left(1 - \frac{2C}{n}\right) \cdot \calO\left(C^3\log n\right) 
+ \calO\left(\frac{2C}{n}\right)\cdot \calO\left(Cn \right) 
= \calO\left(C^3\log n\right)
\end{split}
\end{align*}
\end{proof}
Theorem \ref{thm:nboupper} is a direct consequence.
\begin{proof}
For a source-destination pair $(s,t)$ with distance $d=dist(s,t)$, the expected routing length
of \emph{NextBestOnce} is bound by
\begin{align*}
\E(R^{NBO}(s,t)) 
= \E(R^{NBO}_1(s,t)) + \E(R^{NBO}_2(s,t)) 
= \calO\left(\log^{\alpha-1} d \log \log d + C^3 \log n\right)
\end{align*}
by Lemma \ref{lem:phase1} and \ref{lem:phase2}.
The distance between two nodes is at most $n/2$, so
\begin{align*}
\max_{s,t \in V} \E(R^{NBO}(s,t)) 
= \calO\left(\log^{\alpha-1} n \log \log n + C^3 \log n\right).
\end{align*}
as claimed.
\end{proof}

\section{Proof of Theorem \ref{thm:nbolower}}
\label{sec:lower}

As for the upper bound, the proof is done by dividing the routing process into two phases.
Let $R^{NBO}_1(s,t)$ be the number of nodes contacted to reach a node within distance $C$
and $R^{NBO}_2(s,t)$ the number of steps needed to get from this node to $t$.

\begin{lemma}
\label{lem:phase1lower}
For a graph $G=(V, E) \in  \mathcal{D}(n, 1, C, S_{\alpha})$ and two nodes $s,t \in V$,
the expected routing length for the first phase is
\begin{align*}
\E(R^{NBO}_1(s,t)) = \Omega\left( \log^{\alpha-1} n \right).
\end{align*} 
\end{lemma}
The proof of Lemma \ref{lem:phase1lower} is very similar to the one presented in 
\cite{fraigniaud09effect} and can be found in \cite{roos11analysis}.

In order to show the second result, some facts about $\mathcal{D}(n, 1, C, S_{\alpha})$
are needed. 
Recall that $l(u,v)$ denotes the event that there is a long-range link incident
to $v$ and $u$.

\begin{lemma}
\label{lem:sqrt}
The probability that a long-range link is at least of length $2\sqrt{n}$ is constant, i.e.
\begin{align*}
P(dist(u,v) \geq 2\sqrt{n} | l(u,v)) = \Omega\left(1\right).
\end{align*} 
\end{lemma}
\begin{proof}
First, consider that the probability that $u$ and $v$ is given by the following
(where $\gamma=\Theta(\log n)$ is the normalization constant in Section \ref{sec:model}):
\begin{align*}
\begin{split}
&P\left(l(u,v)|dist(u,v)=d\right)
= \sum_{l_1=1}^\infty \sum_{l_2=1}^\infty \left(1-e^{-\frac{l_1l_2}{d\gamma}}\right) P\left(S_{\alpha}=l_1\right)P\left(S_{\alpha}=l_2\right) \\
&= \sum_{l_1=1}^\infty \sum_{l_2=1}^\infty \Theta\left(\frac{l_1l_2}{d\gamma} P\left(S_{\alpha}=l_1\right)P\left(S_{\alpha}=l_2\right) \right)
= \Theta\left( \frac{1}{\gamma d} \E(S_{\alpha})^2 \right) = \Theta\left( \frac{1}{\gamma d}\right). 
\end{split}
\end{align*}
The last step holds since the expectation of $S_{\alpha}$ is constant. 
Note that
the probability that two randomly selected nodes on a ring of length $n$ have at least distance
$\sqrt{n}$ converges to 1. 
The claim now easily follows:
\begin{align*}
\begin{split}
&P\left(dist(u,v) \geq 2\sqrt{n} | l(u,v)\right) = \frac{P\left(l(u,v) | dist(u,v) \geq 2\sqrt{n}\right)} {P\left(l(u,v)\right)}P\left(dist(u,v) \geq 2\sqrt{n}\right) \\
&=  \Omega\left(\frac{P\left(l(u,v) | dist(u,v) \geq 2\sqrt{n}\right)} {P\left(l(u,v)\right)} \right) 
= \Omega\left( \frac{\sum_{d=2\sqrt{n}}^{n/2} \frac{1}{\gamma d}\cdot\frac{2}{n}}
{\sum_{d=1}^{n/2} \frac{1}{\gamma d}\cdot\frac{2}{n}} \right) 
= \Omega\left( \frac{\sum_{d=2\sqrt{n}}^{n/2} \frac{1}{d}}{\sum_{d=1}^{n/2} \frac{1}{d}} \right) \\
&= \Omega\left( \frac{\log (n/2) - \log (2\sqrt{n})}{\log (n/2)} \right) 
= \Omega\left( \frac{1/2\log n -4}{\log n } \right) 
= \Omega\left(1\right)
\end{split}
\end{align*}
The second last second step follows from $\sum_{i=1}^n \frac{1}{i} = \Theta(\log n)$.
\end{proof}

\begin{lemma}
\label{lem:expQ}
The expected number of nodes $Q$ in $V\setminus B_{\sqrt{n}}(t)$ that have a neighbor in
$B_{d}(t)$ for any $d < \sqrt{n}$ is 
\begin{align*}
E(Q) = \Omega \left(d \right).
\end{align*} 
\end{lemma}
\begin{proof}
The claim follows  from the fact that
$P(l(u,v)|dist(u,v)=d) = \Theta\left( \frac{1}{d\log n }\right)$
for any pair of nodes $(u,v)$.
\begin{align*}
\begin{split}
E(Q) 
&=\sum_{d_1=\sqrt{n}}^{n/2} \sum_{d_2=0}^{d} \Theta\left(\frac{1}{(d_1-d_2)\log n}\right)
+ \sum_{d_1=\sqrt{n}}^{n/2} \sum_{d_2=1}^{d} \Theta\left(\frac{1}{(d_1+d_2)\log n}\right) \\
&= \Omega\left( \sum_{d_1=\sqrt{n}}^{n/2} \sum_{d_2=1}^{d} \frac{2}{2d_1\log n} \right)
= \Omega\left( \sum_{d_1=\sqrt{n}}^{n/2} \frac{2d}{2d_1\log n}\right) 
= \Omega\left( d \right) 
\end{split}
\end{align*}
The last step uses $\sum_{d_1=\sqrt{n}}^{n/2} \frac{1}{d_1\log n} = \Omega(1)$ as shown in the proof
of Lemma \ref{lem:sqrt}.
\end{proof}
We can now derive a lower bound on $R^{NBO}_2$.

\begin{lemma}
\label{lem:phase2lower}
For a graph $G=(V, E) \in  \mathcal{D}(n, 1, C, S_{\alpha})$ with $C < \frac{1}{4}n^{1/4}$ and two nodes $s,t \in V$,
the expected routing length for the second phase is 
\begin{align*}
\E(R^{NBO}_2(s,t)) = \Omega\left(C \right).
\end{align*} 
\end{lemma}
\begin{proof}
The above bound is obtained by showing that $P(R^{NBO}_2(s,t) \geq C/4 | A) = \Omega\left( 1 \right)$
for a suitable event $A$ with $P(A) = \Omega\left( 1 \right)$. It follows directly that 
$\E(R^{NBO}_2(s,t))$ grows at least linearly with $C$.
Recall that $t$ has two local
neighbors $a^t_1$ and $b^t_1$ within distance $C$ of $t$.
The set of short-range neighbors of a node $u$ is denoted by $SN(u)$, whereas $LN(u)$ is
the set of long-range neighbors. Furthermore, we abbreviate $U=LN(t)\cup LN(a^t_1) \cup LN(b^t_1)$.
The event $A = A_1 \cap A_2 \cap A_3 \cap A_4$ is the intersection of the following events:
\begin{itemize} 
\item $A_1 = \{ v \notin \{t,a^t_1,b^t_1\} \}$: the first node within distance $C$ of $t$ is not $t$, $a^t_1$ or $b^t_1$
\item $A_2 = \{SN(t) = 2 \}$: $a^t_1$ and $b^t_1$ are $t$'s only short-range neighbors
\item $A_3 = \{|LN(t)| \leq 1 \} \cap \{ |LN(a^t_1)| \leq 1 \} \cap \{ |LN(b^t_1)| \leq 1 \}$: $t$ as well as its two short-range neighbors have maximally one long-range neighbor
\item $A_4 = \bigcup_{u \in U} \{ dist(u,t) \geq \sqrt{n} \}$: $t$ as well as its short-range neighbors have only long-range neighbors at distance at least $\sqrt{n}$ to $t$
 \end{itemize}
Before showing that $P(A) = \Omega\left( 1 \right)$, note that indeed $P(R^{NBO}_2(s,t) \geq C/4 | A) = \Omega\left( 1 \right)$.
\emph{NextBestOnce} increases the distance to $t$ by at most $C$ in each step, hence by conditioning on $A_3$ (and recalling that $C\cdot C/4 < \sqrt{n}$), $t$, $a^t_1$ and
$b^t_1$ can not be contacted by a long-range neighbor in less than $C/4$ steps.
Therefore, $t$ can only be found in less than $C/4$ steps if a node on the path contacts either $a^t_1$ or $b^t_1$ 
via a short-range link (by event $A_2$ and $A_3$).
The probability that a node $u \in \{a^t_1, b^t_1 \}$ is a short-range neighbor of a node $w$ on the routing path $X$
is 
\begin{align}
\label{eq:plocal}
\begin{split}
&P\left(u \in SN(w) | w \in X\right) 
= P\left(u \in \{a^w_1,b^w_1 \} \cup w \in \{a^u_1,b^u_1 \} | w \in X\right)\\
& \leq 2P\left(u \in \{a^w_1,b^w_1 \} | w \in X\right)  \leq 2P\left(u \in \{a^w_1,b^w_1 \}| dist(u,w)\leq C\right) = \frac{2}{C}.
 \end{split}
\end{align}
The second last step holds because the probability that two nodes are short-range neighbors
is maximal when their distance is at most $C$.

Applying an union bound, the probability that one of the 
first $C/4$ nodes on the path after reaching $v$ has an edge to either $a^t_1$ or $b^t_1$ is bounded by:
\begin{align}
\label{eq:lowbound}
P\left(R^{NBO}_2(s,t) \geq C/4 | A\right) = \Omega\left( \left( 1 - \frac{4}{C}\right)^{C/4} \right) = \Omega\left(1\right)
\end{align}
The last step holds, because $(1-1/x)^x$ converges to $1/e$ for $x \rightarrow \infty$.

It remains to show $P(A) = \Omega\left(1\right)$. Using independence of edge selection,
we can rewrite:
\begin{align*}
P(A) = P(A_1 \cap A_2 \cap A_3 \cap A_4) = P(A_1 | A_2 \cap A_3 \cap A_4)P(A_2)P(A_3)P(A_4|A_3)
\end{align*}
For determining $P(A_1 | A_2 \cap A_3 \cap A_4)$, we first define $W^d_L = \{ w \in V\setminus B_{d}(t): l(w,B_{C}(t)) \}$,
the set of all nodes with long-range links to a node at distance at least $d$ that have links into $B_{C}(t)$.
Similarly, let $W_S = \{ w \in V\setminus B_{C}(t): SN(w)\cap B_{C} (t) \neq \emptyset \}$ be the set of nodes with
short-range links into $B_{C}(t)$. Denote the predecessor of $v$ on the routing path by $w$.
We consider the complement of $A_1$ to derive the desired bound.
\begin{align}
\label{eq:a1}
\begin{split}
&P\left(A_1^{\bot} | A_2 \cap A_3 \cap A_4\right) 
\leq P\left(A_1^{\bot} \cap w \in W^{C}_L | A_2 \cap A_3 \cap A_4\right) 
+ P\left(A_1^{\bot} \cap w \in W_S | A_2 \cap A_3 \cap A_4\right) \\
&\leq P\left(A_1^{\bot} \cap w \in W^{\sqrt{n}}_L | A_2 \cap A_3 \cap A_4\right) 
+ P\left(A_1^{\bot} \cap w \in W_S | A_2 \cap A_3 \cap A_4\right) \\
&= P\left(A_1^{\bot} | A_2 \cap A_3 \cap A_4 \cap w \in W^{\sqrt{n}}_L\right)P\left(w \in W^{\sqrt{n}}_L | A_2 \cap A_3 \cap A_4\right)\\
 &+ P\left(A_1^{\bot} | A_2 \cap A_3 \cap A_4 \cap w \in W_S\right)P\left(w \in W_S | A_2 \cap A_3 \cap A_4\right) \\
&\leq  P\left(A_1^{\bot} | A_2 \cap A_3 \cap A_4 \cap w \in W^{\sqrt{n}}_L\right) + P\left(A_1^{\bot} | A_2 \cap A_3 \cap A_4 \cap w \in W_S\right) \\
&= \calO\left(\frac{3}{C} \right) + \calO\left( \frac{2}{C}\right)
\end{split}
\end{align}
If $w \in W^{\sqrt{n}}_L$, the short-range links of $t$ do not influence $A_1$. 
For this reason, we can drop the condition $A_2$ for the first term in Eq. \ref{eq:a1}.
By Lemma \ref{lem:expQ} there are $\Omega(C)$ nodes in $V\setminus B_{\sqrt{n}(t)}$ having edges into $B_{C}(t)$. Conditioning on $A_3$ and $A_4$, at most three of these long-range links are incidents to $t$, $a^t_1$ and $b^t_1$.
The second summand $\frac{2}{C}$ in Eq. \ref{eq:a1} is derived
as in Eq. \ref{eq:plocal}. Note that $A_3$ and $A_4$ do not influence the
event, given that $w \in W_S$.
Consequently,
\begin{align*}
P(A_1 | A_2 \cap A_3 \cap A_4)
= 1 - \calO\left(\frac{3}{C} + \frac{2}{C}\right)
= \Omega\left(1 - \frac{5}{C}\right)
= \Omega\left(1 \right)
\end{align*} 
$P(A_2)$ corresponds to the probability that none of the $2(C-1)$ potential short-range neighbors
but $a^t_1, b^t_1$ have chosen $t$ as a neighbor. So
\begin{align*}
P(A_2) = \left(1-\frac{1}{C}\right)^{2(C-1)} = \Omega(1). 
\end{align*}
Similarly to Eq. \ref{eq:lowbound}, the last bound follows from $(1-1/x)^{2x} \rightarrow e^{-2}$.
Long-range edges are selected independently, hence
\begin{align*}
\begin{split}
&P(A_3) 
= P(|LN(t)|\leq 1)^3 
= \Omega\left( \left(P\left(|LN(t)|\leq 1| l_t = 1\right)P\left(l_t=1\right)\right)^3 \right) \\
&= \Omega\left( \left(P\left(LN(t)|\leq 1| l_t = 1\right)P\left(S_{\alpha}=1\right)\right)^3 \right) 
= \Omega\left(P\left(|LN(t)|\leq 1| l_t = 1\right)^3 \right)
= \Omega\left(1\right) . 
\end{split}
\end{align*}
The second last step holds since $P(S_{\alpha}=k)\propto k^{\alpha}$. Furthermore,
the last steps follows from Eq. \ref{eq:expectation}, because an expected degree
of 1 implies that the probability of having a degree of at most 1 is at least 1/2.

For calculating $P(A_4|A_3)$, denote the long-range neighbor of $t$, $a^t_1, b^t_1$ by
$t_l$, $a_l$ and $b_l$, respectively. Because $C < \sqrt{n}$, it follows from $dist(a,a_l) > 2 \sqrt{n}$
that $dist(t,a_l) > \sqrt{n}$.
\begin{align*}
\begin{split}
&P(A_4|A_3)  \\
&=P\left(dist(t_l,t)>\sqrt{n}|l(t_l,t)\right)\cdot P\left(dist(a_l,t)>\sqrt{n}|l(a_l,a^t_1)\right)\cdot P\left(dist(b_l,t)>\sqrt{n}|l(b_l,b^t_1)\right)\\
&\geq P\left(dist(t_l,t)>2\sqrt{n}|l(t_l,t)\right)\cdot P\left(dist(a_l,a^t_1)>2\sqrt{n}|l(a_l,a^t_1)\right)\cdot P\left(dist(b_l,t)>2\sqrt{n}|l(b_l,b^t_1)\right) \\
&= P\left(dist(t_l,t)>2\sqrt{n}|l(t_l,t)\right)^3 
\end{split}
\end{align*}
Now $P(A_4|A_3)=  \Omega\left(1\right)$ is a direct consequence from Lemma \ref{lem:sqrt}.
The above results confirm that indeed
\begin{align*}
\begin{split}
P(A)  = P(A_1 | A_2 \cap A_3 \cap A_4)P(A_2)P(A_3)P(A_4|A_3) = \Omega(1\cdot 1 \cdot 1) = \Omega(1).
\end{split}
\end{align*} 

Thus, we have shown that
\begin{align*}
P\left(R^{NBO}_2(s,t) \geq C/4 \right)\geq P\left(R^{NBO}_2(s,t) \geq C/4 | A\right)P(A) = \Omega\left(1\right).
\end{align*}
Consequently, the expectation grows at least linearly in $C$, i.e.
\begin{align*}
\E(R^{NBO}_2(s,t)) = \Omega\left(C \right).
\end{align*}
\end{proof}
Theorem \ref{thm:nbolower} follows from Lemma \ref{lem:phase1lower} and \ref{lem:phase2lower}, because 
\begin{align*}
\E(R^{NBO}(s,t)) = \E(R_1(s,t)) + \E(R_2(s,t)) = \Omega(\log^{\alpha-1} n + C).
\end{align*}

\section{Proof of Theorem \ref{thm:non}}
\label{sec:nonbound}

Fix $1/2 \leq r \leq 1$ and $0\leq k \leq \alpha-2$. The routing is now split in two
phases: $R^{NoN}_1(s,t)$ gives the number of steps needed to get within distance
$e^{\log^r n}$ of $t$. $R^{NoN}_2(s,t)$ is the number of steps to cover the remaining distance.
For the proof, we assume that the maximum value of $S_{\alpha}$ is $\mu=\Theta\left(\log n\right)$. Restricting the degree is obviously a relaxation, which avoids further case distinctions. The result holds for an unbounded maximum degree as well, as presented
in \cite{roos11analysis}.
We show that
\begin{align*}
\begin{split}
&\E(R^{NoN}(s,t)) 
= \E(R^{NoN}_1(s,t)) + \E(R^{NoN}_2(s,t)) \\
&= \calO\left(\log^{\alpha-r - k(3-\alpha)} n \log \log n\right) +\calO\left(\log^{r(\alpha-1)} n \log \log n + C^3 \log n\right).
\end{split}
\end{align*}
The result is then obtained by finding $r_{min}$ and $k_{min}$ to minimize the above bound.

The bound for the second phase can be derived from the routing length of \emph{NextBestOnce}.

\begin{lemma}
\label{lem:phase2non}
For a graph $G=(V, E) \in  \mathcal{D}(n, 1, C, S_{\alpha})$, two nodes $s,t \in V$, and $1/2 \leq r \leq 1$,
the expected routing length of \emph{NextBestOnce-NoN} after reaching a node within distance 
$e^{\log^r n}$ of $t$ is
\begin{align*}
\E(R^{NoN}_2(s,t)) = \calO\left(\log^{r(\alpha-1)} n \log \log n + C^3 \log n\right).
\end{align*} 
\end{lemma}
\begin{proof}
\emph{NextBestOnce-NoN} is in expectation at least as fast as \emph{NextBestOnce}, using
the same procedure, only with additional information.
Let $u$ be the first node on the routing path with $dist(t,u) \leq e^{\log^r n}$. By
Theorem \ref{thm:nboupper}, the expected routing length to get from $u$ to $t$ is: 
\begin{align*}
\begin{split}
\E(R^{NoN}_2(s,t)) 
= \calO(\E(R^{NBO} (u,t)))
&= \calO\left( \log^{\alpha-1} e^{\log^r n} \log \log e^{\log^r n} +  C^3 \log n \right) \\
&= \calO\left( \log^{r(\alpha-1)} n \log \log n +  C^3 \log n \right)
\end{split}
\end{align*}
This proves the claim.
\end{proof}
The first phase of the routing is considerable more work.
In the following, assume $C < e^{log^{1/2} n}$. 
Otherwise, the bound $\calO\left( C^3 \log n \right)$ holds
for both \emph{NextBestOnce} and \emph{NextBestOnce-NoN} by Theorem \ref{thm:nboupper}.
A preliminary Lemma is needed to determine the probability that nodes are adjacent given their labels.

\begin{lemma}
\label{lem:edgeprob}
Consider a node $u$ with $d = dist(u,t) > e^{\log^{1/2} n}$, and a set $W \subset V$, so that $dist(r,t) > d$ for all $w \in W$. 
Denote by $V^{a}_{d'} = \{v \in V: v \in B_{d'}(t), l_v \geq a \}$ the set of all nodes within distance $d' \leq d$ of the destination and label at least $a$.
Furthermore, assume $\frac{|W|M}{(d-d')\log n} < 1/2$.
The probability that $u$ is adjacent to a node in $V^a_{d'}$, conditioned on $l_u$ 
and the absence of edges between $W$ and $B_{d'}(t)$, 
is bounded by
\begin{align*} 
\begin{split} 
P\left(l(u, V^a_{d'}) | l_u = l \cap l(W,B_{d'}(t))^{\bot}\right) = 
 \Omega\left(\frac{l}{\log n}[\log(d+d'-1)-\log(d-d'+1)]a^{2-\alpha}\right).
\end{split}
\end{align*}
\end{lemma}
\begin{proof}
We show that the expected number $Q$ of nodes $v$ in $V^a_{d'}$ that have a link to $u$ is
\begin{align*}
E(Q) = \Omega\left(\frac{l}{\log n}[\log(d+d'-1)-\log(d-d'+1)]a^{2-\alpha}\right).
\end{align*}
Then the probability of $l(u, V^a_{d'})$ satisfies this bound as well. \\
Proof of the last statement: For each $v \in B_{d'}$, the random variable $Q_v$ is 1 if $v \in V^a_{d'}$ and adjacent to $u$.
Otherwise $Q_v$ is 0. For the sum $Q = \sum_{v \in B_{d'}} Q_v$, it holds that 
\begin{align*}
P(Q=1) \geq  1 - e^{-E(Q)} = \Theta\left(E(Q)\right) \textnormal{ for } E(Q) \rightarrow 0
\end{align*} 
The first inequality follows from the fact that the $Q_v$ are independent and
$1-x \leq e^{-x}$ for $x \in [0,1]$:
\begin{align*}
\begin{split}
P(Q = 0) &= \prod_{i=1}^r P(Q_i=0) = \prod_{i=1}^r (1-E(Q_i)) 
 \leq \prod_{i=1}^r e^{-E(Q_i)} = e^{-E(Q)}
\end{split}
\end{align*}

In the following, $E(Q)$ is computed as the sum of probabilities that each node belongs to the set of neighbors
of $u$ within $V^a_{d'}$. Furthermore, note that in case of a one-dimensional ID space (, i.e. a ring), a node with distance
at most $d'-1$ to $t$ has a distance between $d-d'+1$ and $d+d'-1$ to $u$.
\begin{align}
\label{eq:eq1}
\begin{split}
E(Q) 
&= \sum_{i=d-d'+1}^{d+d'-1} P\left(l(u,v) \cap v \in V^a_{d'} | l_u = l \cap l(W,B_{d'}(t))^{\bot} \cap v \in B_{d'}(t) \cap dist(v,u)=i\right) \\
&=\sum_{i=d-d'+1}^{d+d'-1} \sum_{j=a}^{\infty} P\left(l(u,v) \cap l_v=j | l_u=l \cap l(W,B_{d'}(t))^{\bot} \cap v \in B_{d'}(t) \cap dist(v,u)=i \right)\\
&= \sum_{i=d-d'+1}^{d+d'-1} \sum_{j=a}^{\infty} P\left(l_v=j | l_u=l \cap l(W,B_{d'}(t))^{\bot} \cap v \in B_{d'}(t) \cap dist(v,u)=i\right)\\
& \hspace{2.3cm} \cdot P\left(l(u,v) | l_u=l\cap l(W,B_{d'}(t))^{\bot}  \cap v \in B_{d'}(t) \cap dist(v,u)=i \cap l_v=j \right)\\
&:= \sum_{i=d-d'+1}^{d+d'-1} \sum_{j=a}^{\infty} t_1 \cdot t_2
\end{split}
\end{align}
Recall from Eq. \ref{eq:LRtilde} in Section \ref{sec:model} that two nodes $v$, $w$ are adjacent
with probability
\begin{align}
\label{eq:connect}
P(l(v,w)|dist(v,w)=i \cap l_v=l_1 \cap l_w=l_2)
= 1 - e^{-\frac{l_1l_2}{i\gamma }} 
= \Theta\left(\frac{l_1l_2}{i\gamma }\right)
= \Theta\left(\frac{l_1l_2}{i\log n}\right).
\end{align}
The last steps holds due to Eq. \ref{eq:expectation}.
A scale-free degree distribution $S_{\alpha}$ as defined in Eq. \ref{eq:salpha} is used, i.e. the probability that a node $v$ has label $l$ is proportional to $l^{-\alpha}$.
The probability that a node $v \in B_{d'}$  is adjacent to a node in $W$ is obtained
by a simple union bound.
\begin{align*}
\begin{split}
&P\left( l(v,W) | l_v = l \cap dist(t,W) > d \cap v \in B_{d'}(t) \cap dist(v,u)=i\right)\\
&\leq \sum_{w \in W} P\left( l(v,w)| l_v = l \cap dist(t,W) > d \cap v \in B_{d'}(t) \cap dist(v,u)=i\right) \\
&= \sum_{w \in W} P\left( l(v,w)| dist(v,w) > d-d' \cap l_v = l \cap dist(v,u)=i\right)\\
&= \Theta\left( \frac{l|W|}{(d-d')\log n} \right)
\end{split}
\end{align*} 
The last step holds by Eq. \ref{eq:connect}.
Because the expectation $E(S_{\alpha})$ is constant, we get:
\begin{align*}
\begin{split}
&P\left( l(v,W) | dist(t,W) > d \cap v \in B_{d'}(t) \cap dist(v,u)=i\right)\\
&= \sum_{l=1}^{\infty} P\left( l(v,W) | l_v = l \cap dist(t,W) > d \cap v \in B_{d'}(t) \cap dist(v,u)=i\right)
P(l_v=l) \\
&= \Theta\left( \sum_{l=1}^{\infty} \frac{l|W|}{(d-d')\log n} P(l_v=l) \right)\\
&= \Theta\left( \frac{|W|}{(d-d')\log n}  \right)
\end{split}
\end{align*} 
By assumption, $\frac{|W|M}{(d-d')\log n} < 1/2$, and hence both
$P(l(W,v)^{\bot}| v \in B_{d'}(t) \cap dist(v,u)=i) = \Omega\left(1\right)$ and
$P\left(l(W,v)^{\bot} | l_v=j \cap v \in B_{d'}(t) \cap dist(v,u)=i\right)= \Omega\left(1\right)$.
Given that labels and edges are chosen independently, $t_1$ is easily obtained as:
\begin{align*}
\begin{split}
t_1 &= P\left(l_v = j| l(W,v)^{\bot}\cap v \in B_{d'}(t) \cap dist(v,u)=i\right) \\
&= \frac{P\left(l(W,v)^{\bot} | l_v=j \cap v \in B_{d'}(t)\cap dist(v,u)=i\right)P\left(l_v=j \cap v \in B_{d'}(t)\cap dist(v,u)=i\right)}
{P\left(l(W,v)^{\bot}|v \in B_{d'}(t)\right)P\left(v \in B_{d'}(t)\cap dist(v,u)=i\right)} \\
&= \frac{P\left(l(W,v)^{\bot} | l_v=j \cap v \in B_{d'}(t)\cap dist(v,u)=i\right)P\left(l_v=j \right)P\left( v \in B_{d'}(t)\cap dist(v,u)=i\right)}
{P\left(l(W,v)^{\bot}|v \in B_{d'}(t) \cap dist(v,u)=i\right)P\left(v \in B_{d'}(t)\cap dist(v,u)=i\right)} \\
&= \Omega\left(P(l_v = j)\right) 
= \Omega\left(j^{-\alpha}\right)
\end{split}
\end{align*}
Since edges are chosen independently, the event $l(W,B_{d'}(t))^{\bot}$ does not influence $t_2$. 
So $t_2$ is a consequence from Eq. \ref{eq:connect}.
\begin{align*}
\begin{split}
t_2 &=  P\left(l(u,v) | l_u=l \cap l_v=j \cap dist(u,v)=i\right) 
= \Theta\left(\frac{l\cdot j}{i\gamma}\right)
\end{split}
\end{align*}
Replacing $t_1$ and $t_2$ in Eq. \ref{eq:eq1}, we obtain the desired result:
\begin{align*}
\begin{split}
  E(Q) 
&= \sum_{i=d-d'+1}^{d+d'-1} \sum_{j=a}^{\infty} t_1 \cdot t_2 \\
&= \sum_{i=d-d'+1}^{d+d'-1} \sum_{j=a}^{\infty} \Omega\left( \frac{l}{i \log n} j^{1-\alpha} \right) \\
&= \Omega\left( \sum_{i=d-d'+1}^{d+d'-1} \frac{l}{i\log n}  \sum_{j=a}^{\infty} j^{1-\alpha}\right) \\
&= \Omega\left( \sum_{i=d-d'+1}^{d+d'-1} \frac{l}{i\log n}  \int_{a}^{\infty} x^{1-\alpha} dx \right) \\
&= \Omega\left( \sum_{i=d-d'+1}^{d+d'-1} \frac{l}{i \log n}  a^{2-\alpha} \right) \\
&= \Omega\left( \frac{l}{\log n}[\log (d+d'-1) - \log (d-d'+1)]  a^{2-\alpha} \right)
\end{split}
\end{align*}
This shows that the expected number of neighbors and hence the probability to have one neighbor
within the desired set is indeed 
\begin{align}
P\left(l(u, V^a_{d'}) | l_u=l \cap l(W,B_{d'}(t))^{\bot}\right)
=\Omega\left( \frac{l}{\log n}[\log (d+d'+1) - \log (d-d'-1)] a^{2-\alpha} \right)
\end{align}
as claimed.
\end{proof}

In the following, we model the routing process as a sequence $X_1,X_2,\ldots $, such that
$X_i$ gives the distance of the closest neighbor of the $i$-th node on the path
to $t$. 
The distance of the closest neighbor to $t$ decreases in each step
until a node within distance $C$ is reached. This
cannot be guaranteed for the nodes on the actual path. A node at a higher
distance might be chosen if it has neighbors that are close to the destination.
The monotone decrease of the sequence $X_i$ allows us to make use of the
following Lemma:

\begin{lemma}
\label{lem:xi}
If $X_0,X_1,...$ is a non-negative, integer-valued random process with $X_0 >\lambda \geq 0$, such that for all $d$ with $\lambda < d \leq X_0$
\begin{align*}
P\left(X_{i+2} \leq d/2 |X_0,\ldots ,X_i = j\right) = \Omega\left(\frac{\log d}{\rho}\right) 
\end{align*}
then the expected number of steps until the random process reduces to at most $\lambda $ is $\calO(\rho \log \log X_0)$
\end{lemma}
A proof can be found in \cite{fraigniaud09effect}, Lemma 5.2.

Let $R_i$ denotes the set of all nodes on the path before the $i$-th node and their neighbors.
All events need to be conditioned
on the fact that no node within distance $d=X_i$ has a link to a node in $R_i$, i.e. the event $l(B_d,R_i)^{\bot}$.
The next result is the main part of the proof enabling the use of Lemma \ref{lem:xi} with $\lambda = e^{\log^r n}$, $\rho = \log^{\alpha-r-k(3-\alpha)} n$.

\begin{lemma}
\label{lem:step}
Let $X_i$ be the distance of the closest neighbor of the $i$-th node on the routing path, $1/2 \leq r \leq 1$, 
$0 \leq k \leq \alpha-2$, and $|R_i| < 1/2\sqrt{d}\log n$. The chance that $X_i$ is halved in the next 
two steps is:
\begin{align*}
P\left(X_{i+2} \leq \frac{d}{2} | R_i \cap X_i = d\right) = \Omega\left(\frac{\log d\cdot\log^{r+k(3-\alpha)} n }{\log^\alpha n}\right)
\end{align*}
\end{lemma}
\begin{proof}
Let $u$ be the $i$-th node
on the path. We show the result by distinguishing two cases: $l_u < \log^k n$ and $l_u \geq \log^k n$. 
But before, a case-independent observation is made.

Note that though the distance of a neighbor
of $u$ to $t$ is known, the distance $\Delta$ of $u$ is not given. We bound all the
following probabilities on the event $G = \{d+\sqrt{d} \leq \Delta \leq 2d \} $
The first inequality is necessary to apply Lemma \ref{lem:edgeprob} with
$\frac{|R_i|}{(d+\sqrt{d}-d)\log n} <1/2$. 
The bound $\Delta \leq 2d$ ensures that the $dist(u,t)$ needs to be maximally
quartered to have $X_{i+2} \leq d/2$.
For a lower bound on the event $A$ of halving the distance, $P(A)\geq P(A|G)P(G)$ can be applied.
If $P(G) = \Omega(1)$, $P(A)=\Omega\left( P(A|G)\right)$ holds.
It remains to show $P(G) = \Omega\left(1\right)$. 
The lower bound $\Delta \geq d+ \sqrt{d}$ holds with probability $\Theta\left(1\right)$ 
by Lemma \ref{lem:sqrt}.
The upper bound $\Delta \leq 2d$ holds with probability $\Omega\left(1\right)$
as well, as can be seen from the proof of Theorem 2.4 in \cite{fraigniaud09effect}: The probability
that an arbitrary node has a neighbor at half its distance to the destination is shown to be
$\calO\left(\frac{1}{\log^{\epsilon} n}\right)$ for some $\epsilon > 0$. Thus, the probability of not having such a neighbor
is $\Omega\left(1-\frac{1}{\log^{\epsilon} n}\right) = \Omega\left(1\right)$, because $\frac{1}{\log^{\epsilon} n} < 1/2$
for $n$ big enough.

This concludes our case-independent observation, ensuring that indeed $P(G) = \Omega(1)$.
In both cases, $l_u < \log^k n$, and $l_u \geq \log^k n$, we first describe
an event leading to halving the distance, before formally deriving
the probability of the respective event.

Assume $l_u < \log^k n$. The following events result in $X_{i+2}\leq d/2$:
\begin{itemize}
\item a neighbor $v \in B_\Delta (t)$ of $u$ has label $l_v \geq \log^k n$ .
\item $v$ has a neighbor $w \in B_\Delta (t)$ with label $l_w \geq \log n$
\item $w$ has a link into $B_{d/2}$
\item $v$ is the node $u$ chooses as the next node on the routing path, denote this event by $\left\{ Z = v \right\}$
\end{itemize}
All events are conditioned on $F=l(B_d,R_i)^{\bot}) \cap l_u \leq \log^k n \cap G$.
Formally, the probability is determined by:
\begin{align}
\label{eq:logk}
\begin{split}
  &P\left(l(u, V^{\log^k n}_{\Delta}) \cap l(v,V^{\log n}_{\Delta}) \cap l(w,B_{d/2})\cap Z = v | F\right) \\
  &=P\left(l(u, V^{\log n}_{\Delta}) | F\right)\cdot  
  P\left(l(v,V^{\log n}_{\Delta})| l(u, V^{\log^k n}_{\Delta}) \cap F\right)  \\
  & \hspace{0.2cm} \cdot P\left(l(w,B_{d/2}) | l(v,V^{\log n}_{\Delta}) \cap l(u, V^{\log^k n}_{\Delta})\cap F\right) \\  
  & \hspace{0.2cm} \cdot P\left(Z = v | l(w,B_{d/2}) \cap l(v,V^{\log n}_{\Delta})\cap l(u, V^{\log^k n}_{\Delta})\cap F\right) \\
  & := q_1q_2q_3q_4
\end{split}
\end{align}
We now subsequently bound $q_1$, $q_2$, $q_3$ and $q_4$.
$q_1$ can be derived using Lemma \ref{lem:edgeprob} with $d=d'=\Delta$ and the fact that the probability of having a link
is minimal for a node $u$ with $l_u=1$.
\begin{align*}
\begin{split}
q_1 &= P\left(l(u, V^{\log^k n}_{\Delta})  | l(B_d,R_i)^{\bot} \cap l_u \leq \log^k n \cap G \right) \\
&\geq P\left(l(u, V^{\log^k n}_{\Delta})  | l(B_d,R_i)^{\bot} \cap l_u=1 \cap G \right) \\
&=\Omega\left(\frac{1}{\log n}[\log (2\Delta-1)-0]\log^{k(2-\alpha)} n \right) \\
&=\Omega\left(\frac{\log d}{\log n}\log^{k(2-\alpha)} n \right) 
\end{split}
\end{align*}
The last step uses $d \leq \Delta \leq 2d$.
Since links are selected independently, the events $l(u, V^{\log^k n}_{\Delta})$
and $l(v,V^{\log n}_{\Delta})$ are independent, but $v \in V^{\log^k n}_{\Delta}$
influences $l(v,V^{\log n}_{\Delta})$.
The maximal distance $v$ can have to $t$ is $\Delta $.
Because labels are selected independently $l_u \leq \log^k n$ does not influence $q_2$ or $q_3$.
Hence $q_2$ is derived similarly to $q_1$:
\begin{align*}
\begin{split}
q_2 &= P\left(l(v, V^{\log n}_{\Delta})  | l(B_d,R_i)^{\bot} \cap l_v \geq \log^k n \cap G\right) \\
&\geq P\left( l(v, V^{\log n}_{\Delta})  | l(B_d,R_i)^{\bot} \cap l_v=\log^k n \cap G\right) \\
&=\Omega\left(\frac{\log^k n}{\log n}[\log (2\Delta-1)-0]\log^{(2-\alpha)} n \right) \\
&=\Omega\left(\frac{\log d}{\log n}\log^{k} n \log^{(2-\alpha)} n \right) 
\end{split}
\end{align*}
Note that $B_{d/2} = V^1_{d/2}$, so Lemma \ref{lem:edgeprob} is applied
to determine $q_3$ as well.
Furthermore, the function $\log (x+d/2) - \log (x-d/2)$, being a monotone decreasing function
for $x > d/2$, assumes its maximum in the interval $[d,2d]$ at $\Delta=d$.
\begin{align*}
\begin{split}
q_3 &= P\left(l(w, B_{d/2})  | l(B_d,R_i)^{\bot} \cap l_w \geq \log n\cap G\right) \\
&\geq P\left(l(w, B_{d/2})  | l(B_d,R_i)^{\bot} \cap l_w=\log n\cap G\right) \\
&=\Omega\left(\frac{\log n}{\log n}[\log (\Delta+d/2-1)-\log (\Delta-d/2+1)] \right) \\
&=\Omega\left(\log (d+d/2-1) - \log(d-d/2+1) \right) \\
&= \Omega\left(1\right)
\end{split}
\end{align*}
We are considering the case when u has  $\calO\left( \log^k n \right)$ neighbors at distance at least d from t. Note that the probability 
that one of $\log^k n$ arbitrary nodes link to a certain node is asymptotically the same 
as considering one node with degree $\log^k n$. 
Hence the probability that $v$ or any one of $u$'s remaining neighbors has the closest neighbor to $t$ are of the same order, i.e. $q_4=\Omega\left(1\right)$.

Combining the results for the individual terms, we get a bound for halving the distance
in case of $l_u \leq \log^k n$.
\begin{align*}
\begin{split}
&P(X_{i+2} \leq \frac{d}{2} | R_i \cap X_i = d) \\
&= \Omega\left(q_1q_2q_3q_4\right) \\
&= \Omega\left(\frac{\log d}{\log n}\log^{k(2-\alpha)} n \cdot \frac{\log d}{\log n}\log^{k} n \log^{(2-\alpha)} n\right) \\ 
&= \Omega\left(\frac{\log d \log^{k(3-\alpha)} n }{\log^\alpha n} \log d \right) \\
&= \Omega\left(\frac{\log d \log^{k(3-\alpha)} n}{\log^\alpha n} \log^r n \right)
\end{split}
\end{align*}
The last step uses that $\log d > \log e^{\log^r n} = \log^r n$.
Having shown the result for the case $l_u < \log^k n$, $l_u \geq \log^k n$ has to be treated differently,
because $q_4$ cannot be bounded as above. However, this case is easier, since one does not need to contact a
node with label at least $\log^k n$ first.
Here we consider the following events:
\begin{itemize}
\item $u$ has a neighbor $w$ within distance $\Delta$ of $t$ with degree at least $\log n$
\item $w$ links into $B_{d/2}$
\end{itemize}
Since $u$ knows the identifiers of $w$'s neighbors, it will select $w$, unless there is some other node being
both a neighbor to $u$ and a node in $B_{d/2}$. 
This corresponds to the second and third
events in case $l_u \leq \log^k n$ and are already bounded by $q_2$ and $q_3$.  
Formally, this event can be written as follows:
\begin{align*}
\begin{split}
  &P\left(l(u, V^{\log n}_{\Delta}) \cap l(w,B_{d/2}) | l_u \geq \log^k n \cap  l(B_d,R_i)^{\bot}\cap G\right) \\
&= P\left(l(u, V^{\log n}_{\Delta}) | l_u \geq \log^k n \cap  l(B_d,R_i)^{\bot}\cap G\right) \\
&\hspace{0.3cm}\cdot 
  P\left(l(w,B_{d/2}) | l(u, V^{\log n}_{\Delta}) \cap l_u \geq \log^k n \cap  l(B_d,R_i)^{\bot}\cap G\right) \\
&= q_2q_3 
= \Omega\left(\frac{\log d}{\log n}\log^{k} n \log^{(2-\alpha)} n \right) 
\end{split}
\end{align*} 
So, we can half the distance in one step with probability 
$\Omega\left(\frac{\log d}{\log n}\log^{k} n \log^{(2-\alpha)} n \right)$.
The sequence $X_i$ is decreasing and $q_1q_4 \leq 1$, so in case $l_u \geq \log^k n$,
it holds that
\begin{align*}
\begin{split}
&P(X_{i+2} \leq \frac{d}{2} | R_i \cap X_i = d) 
= \Omega\left(\frac{\log d \log^{k(3-\alpha)} n}{\log^\alpha n} \log^r n \right)
\end{split}
\end{align*}
as well. This completes the proof.
\end{proof}
Based on Lemma \ref{lem:xi} and Lemma \ref{lem:step}, $\E(R^{NoN}_1(s,t))$
can be bounded.

\begin{lemma}
\label{lem:phase1non}
For a graph $G=(V, E) \in  \mathcal{D}(n, 1, C, S_{\alpha})$ with $C < e^{\log^{1/2} n}$, two nodes $s,t \in V$, $1/2 \leq r \leq 1$ and $0 \leq k \leq \alpha-2$,
the expected routing length of \emph{NextBestOnce-NoN} to reach a node within distance 
$e^{\log^r n}$ of $t$ is
\begin{align}
\label{eq:non1}
\E(R^{NoN}_1(s,t)) = \calO\left(\log^{\alpha-r-k(3-\alpha)} n \log \log n \right).
\end{align} 
\end{lemma}
\begin{proof}
By Lemma \ref{lem:step} the probability to half the distance during the next two steps is given by
\begin{align*}
P(X_{i+2} \leq \frac{d}{2} | X_1,X_2,\ldots, X_i = d) = \Omega\left(\frac{\log d\cdot\log^{r+k(3-\alpha)} n }{\log^\alpha n}\right)
= \Omega\left(\frac{\log d}{\log^{\alpha-r-k(3-\alpha)} n}\right)
\end{align*}  
as long as $d > e^{\log^r n}$ and $|R_i|  < \sqrt{d}/2$. The later holds with probability at least
$1- \frac{1}{n}$, as can be seen from the proof for the upper bound of \emph{NextBestOnce} (see
\cite{roos11analysis} or for a similar argumentation \cite{fraigniaud09effect}, Theorem 2.4). It is
shown that \emph{NextBestOnce} needs at most $\calO\left(\log^3 n\right)$ steps with probability $\Omega\left( 1 - 1/n \right)$.
Since we assume the maximal degree to be bounded logarithmically, $|R_i| \leq K\log^4 n$ for some
constant $K$ follows.
Hence, Lemma \ref{lem:xi} with $\rho=\log^{\alpha-r-k(3-\alpha)} n$ and $\lambda=e^{\log^r b}$ can be applied to obtain
Eq. \ref{eq:non1}:
\begin{align*}
&\E(R^{NoN}_1(s,t)) \\
&= P(|R_i| \leq K\log^4 n) \E(R^{NoN}_1(s,t)||R_i| \leq K \log^4 n) \\
&+ (1-P(|R_i| \leq K\log^4 n))\E(R^{NoN}_1(s,t)| |R_i| > K \log^4 n)\\
&= \Omega\left( 1 - \frac{1}{n} \right) \calO\left( \log^{\alpha-r-k(3-\alpha)} n \log \log n \right) 
+ \calO\left( \frac{1}{n} \right)\calO\left( n \right) \\
&=  \calO\left(\log^{\alpha-r-k(3-\alpha)} n \log \log n  \right)
\end{align*}
In the second last step holds since $e^{\log^r n} > C$, so the distance is guaranteed to decrease
in each step, and thus maximally $n/2 - e^{\log^r n}$ are needed to complete the first phase.
\end{proof}
Theorem \ref{thm:non} can now be shown solving a two-dimensional extremal value problem.
\begin{proof}
It follows from Lemma \ref{lem:phase1non} and \ref{lem:phase2non} that for all $(k,r)\in [0,\alpha-2]\times [1/2,1]$
\begin{align*}
\E(R^{NoN}(s,t)) = 
\calO\left(\log^{\alpha-r - k(3-\alpha)} n \log \log n + \log^{r(\alpha-1)} n \log \log n + C^3 \log n\right)
\end{align*}
We need to find $(k_{min}, r_{min})$ such that
\begin{align}
f(k,r) = \log^{\alpha-r - k(3-\alpha)} n  + \log^{r(\alpha-1)} n
\end{align}
is minimized.
Computing the gradient of f gives:
\begin{align*}
\begin{split}
&Df 
 = \begin{pmatrix} -(3-\alpha)\log \log n(\log n)^{\alpha-r_{min}-k_{min}(3-\alpha)} \\
-\log \log n(\log n)^{\alpha-r_{min}-k_{min}(3-\alpha)} + (\alpha-1)\log \log n(\log n)^{r_{min}(\alpha-1)} \end{pmatrix} 
  \neq \begin{pmatrix} 0 \\ 0 \end{pmatrix}
\end{split}
\end{align*}
So the f takes its minimum on the border of the \mbox{$[0, \alpha-2]x[1/2,1]$}, e.g. if either $k_{min}=0$, $k_{min}=\alpha-2$, $r_{min}=0.5$ or $r_{min}=1$. 
When $k_{min}=0$, a node of degree at least one needs to be contacted first. This leads essentially to the same scenario used to obtain the bound for \emph{NextBestOnce}, and cannot have an improved complexity.
The same goes for the case $r_{min}=1$, because $e^{\log^1 n} = n$, so only the second phase, for which the complexity is bounded by that of \emph{NextBestOnce} is considered. As for $r_{min}=0.5$, observe the exponent of the first summand of f.
\begin{align*}
\alpha - 0.5 - k\cdot (3-\alpha) \geq \alpha - 0.5 - (\alpha-2)\cdot (3-\alpha) > \alpha-1
\end{align*}
The last step uses that $2 < \alpha < 3$.
So, an improved bound with regard to \emph{NextBestOnce} can only obtained for $k_{min}=\alpha-2$.
We determine $r_{min}$ by minimizing
$$g(r) = (\log n)^{\alpha - (\alpha-2)(3-\alpha)-r} + (\log n)^{r(\alpha-1)}$$
The first derivative of g is
\begin{align*}
\begin{split}
g'(r) &= -\log \log n (\log n)^{\alpha - (\alpha-2)(3-\alpha)-r} 
 + (\alpha-1)\log \log n (\log n)^{r(\alpha-1)}
\end{split}
\end{align*}
Setting $g'(r_{min}) = 0$ we get that 
\begin{align*}
(\alpha-1)\log \log n (\log n)^{r_{min}(\alpha-1)} 
&= \log \log n (\log n)^{\alpha - (\alpha-2)(3-\alpha)-r_{min}} \\
 (\log n)^{r_{min}(\alpha-1)+r_{min}} &= \frac{1}{\alpha-1}(\log n)^{\alpha - (\alpha-2)(3-\alpha)} \\
\left((\log n)^{r_{min}}\right)^{\alpha} &= \frac{1}{(\alpha-1)}(\log n)^{\alpha - (\alpha-2)(3-\alpha)} \\
(\log n)^{r_{min}} &= \frac{1}{(\alpha-1)^{1/\alpha}}(\log n)^{1-\frac{(\alpha-2)(3-\alpha)}{\alpha}} \\
\end{align*}
Finally, we get
\begin{align}
\label{eq:rmin}
\begin{split}
r_{min} &= \frac{\log \frac{1}{(\alpha-1)^{1/\alpha}} + \left(1-\frac{(\alpha-2)(3-\alpha)}{\alpha}\right) \log \log n}{\log \log n} \\
&= \frac{\log \frac{1}{(\alpha-1)^{1/\alpha}}}{\log \log n} + \left( 1-\frac{(\alpha-2)(3-\alpha)}{\alpha} \right)
\end{split}
\end{align}
This is indeed a minimum since
\begin{align*}
\begin{split}
&g''(r_{min}) = 
(\log \log n)^2 \left( (\log n)^{\alpha - (\alpha-2)(3-\alpha)-r_{min}} + (\alpha-1)^2 (\log n)^{r_{min}(\alpha-1)} \right) 
 > 0
\end{split}
\end{align*}
Consider that for $a = \log \frac{1}{(\alpha-1)^{1/\alpha}}$
\begin{align*}
(\log n)^{\frac{a}{\log n}} = 2^{\frac{a \log n}{\log n}} = 2^{a}.
\end{align*}
By this, the first summand in Eq. \ref{eq:rmin} does not contribute to the asymptotic
complexity, allowing us to use only the second summand
\begin{align*}
r^*_{min} = 1-\frac{(\alpha-2)(3-\alpha)}{\alpha}
\end{align*}
for the routing bound in Theorem \ref{thm:non}.
The upper bound on \emph{NextBestOnce-NoN} is then obtained as 
\begin{align*}
\begin{split}
&\calO\left(f(k_{min},r*_{min}) \log \log n + C^3 \log n \right) 
= \calO\left( \log^{\delta(\alpha)(\alpha-1)} n \log \log n + C^3 \log n \right)
\end{split}
\end{align*}
for 
$\delta(\alpha)=1-\frac{(\alpha-2)(3-\alpha)}{\alpha}$. 
This completes
the remaining steps in the proof of Theorem \ref{thm:non}.
\end{proof}
\section{Conclusion and Future Work}
\label{sec:conclusion}


We have provided an analysis of \emph{NextBestOnce},
an alternative routing algorithm based solely on
information about direct neighbors with guaranteed convergence
in any embedded graph.
In the context of a model for heuristically embedded social graphs, the expected routing
length of \emph{NextBestOnce} can be bound in terms
of the number $n$ of participants, the exponent $\alpha$
of the scale-free degree distribution of the social graph,
and a parameter $C$ measuring the accuracy of the embedding:
\emph{NextBestOnce}'s
expected routing length satisfies the upper and lower bounds
$\calO\left(\log^{\alpha-1} n \log \log n + C^3 \log n\right)$ and $\Omega\left(\log^{\alpha-1} n
+C\right)$.
By this, \emph{NextBestOnce} achieves polylog performance as long as $C$
is bound polylog.
This result complements our earlier work, which shows
that currently deployed algorithms do not achieve polylog performance \cite{roos13contribution},
even in case of constant $C$.
Furthermore, we have shown that using information about the two-hop neighborhood
indeed achieves an asymptotically decreased routing length for sufficiently
accurate embeddings: The extended
algorithm \emph{NextBestOnce-NoN} needs 
$\calO\left(\log^{\delta(\alpha)(\alpha-1)} n \log \log n + C^3 \log n\right)$
hops for $\delta(\alpha) < 1$.


The price of the increased performance are additional local computation, storage,
and maintenance costs due to the increased number of identifiers 
considered for routing decisions.
The expected number of two-hop neighbors is given by 
$E(S^2_{\alpha}) = \Theta\left(\mu^{3-\alpha}\right)$,
where $\mu$ is the maximal degree, whereas the expected
number of neighbors $E(S_{\alpha})$ is bound by a constant.
In other words, the costs per hop are constant when using only
direct neighborhood information, but increase with
the maximal degree when considering the two-hop
neighborhood.
Note that a logarithmic maximal degree is sufficient for 
the proof of Theorem \ref{thm:non}. 
Assuming a logarithmic maximal degree, the expected costs are polylog in the
number of participants, as is common in other structured overlays such as DHTs.
With that in mind, the additional costs seem a reasonable price for the
significantly shorter routes \emph{NextBestOnce-NoN} offers, especially
when considering that routing takes often hundreds of hops in
current Darknet implementations.
 

Since our aim was to show the superiority of \emph{NextBestOnce-NoN}
in comparison to \emph{NextBestOnce}, we did not provide
a lower bound on the performance of \emph{NextBestOnce-NoN}.
NoN routing has been shown to be optimal in a similar context \cite{Manku04NoN},
in as far as that the expected routing length is asymptotically equal
to the diameter of the graph. It remains to be seen if the result
holds in case of scale-free degree distributions as well.
In addition, we plan to analyze the dependence of routing length and
the accuracy of the embedding in more detail, aiming to close
the gap between the linear lower and the at least cubic upper bound.

\bibliographystyle{plain}
\bibliography{tech}

\begin{thebibliography}{10}

\bibitem{buchegger09peerson}
Sonja Buchegger, Doris Schi{\"{o}}berg, Le~Hung Vu, and Anwitaman Datta.
\newblock {PeerSoN: P2P Social Networking}.
\newblock In {\em Social Network Systems}, 2009.

\bibitem{Chaintreau08networks}
Augustin Chaintreau, Pierre Fraigniaud, and Emmanuelle Lebhar.
\newblock Networks become navigable as nodes move and forget.
\newblock In {\em Proceedings of the 35th international colloquium on Automata,
  Languages and Programming, ICALP '08}, 2008.

\bibitem{ClarkeSTV10}
Ian Clarke, Oskar Sandberg, Matthew Toseland, and Vilhelm Verendel.
\newblock Private communication through a network of trusted connections: The
  dark freenet.
\newblock http://freenetproject.org/papers.html, 2010.

\bibitem{ClarkeSWH00}
Ian Clarke, Oskar Sandberg, Brandon Wiley, and Theodore~W. Hong.
\newblock Freenet: A distributed anonymous information storage and retrieval
  system.
\newblock In {\em International Workshop on Design Issues in Anonymity and
  Unobservability}, 2000.

\bibitem{Coppersmith02diameter}
Don Coppersmith, David Gamarnik, and Maxim Sviridenko.
\newblock The diameter of a long-range percolation graph.
\newblock {\em Random Struct. Algorithms}, 21(1), 2002.

\bibitem{cutillo09privacy}
Leucio-Antonio Cutillo, Refik Molva, and Thorsten Strufe.
\newblock {Privacy Preserving Social Networking Through Decentralization}.
\newblock In {\em {6th International Conference on Wireless On-demand Network
  Systems and Services (WONS)}}, pages 145 -- 152, 2009.

\bibitem{CvetkovskiCrovella09}
Andrej Cvetkovski and Mark Crovella.
\newblock Hyperbolic embedding and routing for dynamic graphs.
\newblock In {\em Proceedings of the 28th IEEE International Conference on
  Computer Communications, INFOCOM '09}, 2009.

\bibitem{DellAmico07}
Matteo Dell'Amico.
\newblock Mapping small worlds.
\newblock In {\em Proceedings of the 7th International Conference on
  Peer-to-Peer Computing, P2P '07}, 2007.

\bibitem{EppsteinGoodrich11-SuccinctHyperEmbed}
David Eppstein and Michael~T. Goodrich.
\newblock Succinct greedy geometric routing using hyperbolic geometry.
\newblock {\em IEEE Trans. Computers}, 60(11):1571--1580, 2011.

\bibitem{EvansGrothoff11-R5N}
Nathan~S. Evans and Christian Grothoff.
\newblock {R5N}: Randomized recursive routing for restricted-route networks.
\newblock In {\em Proceedings of the 5th International Conference on Network
  and System Security, NSS '11}, 2011.

\bibitem{FluryEtAl09}
Roland Flury, Sriram~V. Pemmaraju, and Roger Wattenhofer.
\newblock Greedy routing with bounded stretch.
\newblock In {\em Proceedings of the 28th IEEE International Conference on
  Computer Communications, INFOCOM '09}, 2009.

\bibitem{fraigniaud09effect}
Pierre Fraigniaud and George Giakkoupis.
\newblock The effect of power-laws on the navigability of small worlds.
\newblock In {\em Proceedings of the 23rd annual ACM symposium on Principles of
  distributed computing, PODC '09}, 2009.

\bibitem{Fraigniaud10searchability}
Pierre Fraigniaud and George Giakkoupis.
\newblock On the searchability of small-world networks with arbitrary
  underlying structure.
\newblock In {\em Proceedings of the 42nd ACM symposium on Theory of
  computing}, STOC '10, 2010.

\bibitem{Giakkoupis11optimal}
George Giakkoupis and Nicolas Schabanel.
\newblock Optimal path search in small worlds: dimension matters.
\newblock In {\em Proceedings of the 43rd Symposium on Theory of Computing,
  STOC '11}, 2011.

\bibitem{HerzenEtAl11}
Julien Herzen, C{\'e}dric Westphal, and Patrick Thiran.
\newblock Scalable routing easy as pie: A practical isometric embedding
  protocol.
\newblock In {\em Proceedings of the 19th IEEE International Conference on
  Network Protocols, ICNP '11}, 2011.

\bibitem{IsdalEtAl10-OneSwarm}
Tomas Isdal, Michael Piatek, Arvind Krishnamurthy, and Thomas~E. Anderson.
\newblock Privacy-preserving p2p data sharing with oneswarm.
\newblock In {\em Proceedings of the ACM SIGCOMM 2010 conference, SIGCOMM '10},
  2010.

\bibitem{kleinberg00small}
Jon Kleinberg.
\newblock The small-world phenomenon: An algorithmic perspective.
\newblock In {\em Proceedings of the 32nd Symposium on Theory of Computing,
  STOC '00}, 2000.

\bibitem{Kleinberg07}
Robert Kleinberg.
\newblock Geographic routing using hyperbolic space.
\newblock In {\em Proceedings of the 26th IEEE International Conference on
  Computer Communications, INFOCOM '07}, 2007.

\bibitem{Lebhar04almost}
Emmanuelle Lebhar and Nicolas Schabanel.
\newblock Almost optimal decentralized routing in long-range contact networks.
\newblock In {\em Proceedings of the 30th international colloquium on Automata,
  Languages and Programming, ICALP '04}, 2004.

\bibitem{Manku04NoN}
Gurmeet~Singh Manku, Moni Naor, and Udi Wieder.
\newblock Know thy neighbor's neighbor: the power of lookahead in randomized
  p2p networks.
\newblock In {\em Proceedings of the 36th annual ACM symposium on Theory of
  computing, STOC '04}, 2004.

\bibitem{Martel03thecomplexity}
C.~Martel and V.~Nguyen.
\newblock The complexity of message delivery in kleinberg’s small-world
  model.
\newblock Technical report, UC Davis Department of Computer Science, 2003.

\bibitem{Martel04analyzing}
Chip Martel and Van Nguyen.
\newblock Analyzing kleinberg's (and other) small-world models.
\newblock In {\em Proceedings of the 33rd annual ACM symposium on Principles of
  distributed computing, PODC '04}, 2004.

\bibitem{Maymounkov06}
Petar Maymounkov.
\newblock Greedy embeddings, trees, and euclidean vs. lobachevsky geometry.
\newblock
  \url{https://www.pdos.lcs.mit.edu/~petar/papers/maymounkov-greedy-prelim.pdf},
  2006.

\bibitem{MittalEtAl12-XVine}
Prateek Mittal, Matthew Caesar, and Nikita Borisov.
\newblock X-vine: Secure and pseudonymous routing using social networks.
\newblock In {\em Proceedings of the 19th Annual Network \& Distributed System
  Security Symposium, NDSS '12}, 2012.

\bibitem{PopescuEtAl06-Turtle}
Bogdan~C. Popescu, Bruno Crispo, and Andrew~S. Tanenbaum.
\newblock Safe and private data sharing with turtle: Friends team-up and beat
  the system.
\newblock In {\em Proceedings of the 12th International Workshop Security
  Protocols}. Springer, 2006.

\bibitem{roos11analysis}
Stefanie Roos.
\newblock Analysis of routing in sparse small-world topologies.
\newblock Diplomarbeit, TU Darmstadt, 2011.

\bibitem{roos12provable}
Stefanie Roos and Thorsten Strufe.
\newblock Provable polylog routing for darknets.
\newblock In {\em Proceedings of the 4th Workshop on Hot Topics in Peer-to-peer
  Computing and Online Social Networking, HotPOST '12}, 2012.

\bibitem{roos13contribution}
Stefanie Roos and Thorsten Strufe.
\newblock A contribution to darknet routing.
\newblock In {\em Proceedings of the 32nd IEEE International Conference on
  Computer Communications, INFOCOM '13}, 2013.

\bibitem{Sandberg06}
Oskar Sandberg.
\newblock Distributed routing in small-world networks.
\newblock In {\em Proceedings of the 8th Workshop on Algorithm Engineering and
  Experiments, ALENEX '06}, 2006.

\bibitem{SchillerEtAl11-LCM}
Benjamin Schiller, Stefanie Roos, Andreas H\"{o}fer, and Thorsten Strufe.
\newblock Attack resistant network embeddings for darknets.
\newblock In {\em Proceedings of the 30th Symposium on Reliable Distributed
  Systems Workshops, SRDSW '11}, 2011.

\bibitem{VassermanEtAl09-MCON}
Eugene Vasserman, Rob Jansen, James Tyra, Nicholas Hopper, and Yongdae Kim.
\newblock Membership-concealing overlay networks.
\newblock In {\em Proceedings of the 17th ACM conference on Computer and
  communications security, CCS '09}, 2009.

\bibitem{WestphalPei09}
C{\'e}dric Westphal and Guanhong Pei.
\newblock Scalable routing via greedy embedding.
\newblock In {\em Proceedings of the 28th IEEE International Conference on
  Computer Communications, INFOCOM '09}, 2009.

\end{thebibliography}

\end{document}